\documentclass[11pt]{article}
\usepackage[margin=1in]{geometry}
\usepackage[T1]{fontenc}
\usepackage[utf8]{inputenc}
\usepackage[all=normal,bibliography=tight]{savetrees}

\usepackage{listings}
\usepackage{cite}
  \usepackage{mathrsfs}

  \usepackage{amstext,amsfonts,amsthm,amsmath,amssymb}

\usepackage{graphicx,tikz}
\usepackage{comment}
\usepackage{url}
\usepackage{xspace}
\usepackage{todonotes}
\usepackage[shortlabels]{enumitem}
\usepackage[absolute]{textpos}
\usepackage{thm-restate}
\usepackage{csquotes}

\def\cqedsymbol{\ifmmode$\lrcorner$\else{\unskip\nobreak\hfil
\penalty50\hskip1em\null\nobreak\hfil$\lrcorner$
\parfillskip=0pt\finalhyphendemerits=0\endgraf}\fi} 

\newcommand{\cqed}{\renewcommand{\qed}{\cqedsymbol}}

\newtheorem{lemma}{Lemma}[section]
\newtheorem{proposition}[lemma]{Proposition}

\newtheorem{theorem}[lemma]{Theorem}

\theoremstyle{definition}

\newcommand{\Oh}{\mathcal{O}}
\newcommand{\pmc}{\Omega}
\newcommand{\cc}{\mathtt{cc}}
\newcommand{\N}{\mathbb{N}}

\newcommand{\ourclass}{\mathcal{C}}

\newcommand{\coverfam}{\mathcal{A}}
\newcommand{\lexlt}{<_{\mathrm{lex}}}
\newcommand{\indlt}{\prec}
\newcommand{\dpres}{\Upsilon}
\newcommand{\States}{\mathbf{States}}
\newcommand{\sepcont}{\widehat{S}}

\newcommand{\weight}{\mathfrak{w}}

\usepackage{todonotes}

\title{%
Induced subgraphs of bounded treewidth and the container method%
\thanks{M. Chudnovsky is supported by NSF grant DMS-1763817. This material is based upon work supported in part by the U. S. Army Research Office under grant number  W911NF-16-1-0404. P. Rz\k{a}\.zewski is supported by Polish National Science Centre grant no. 2018/31/D/ST6/00062. P. Seymour is supported by AFOSR grant A9550-19-1-0187 and NSF grant DMS-1800053.
  This research is a part of a project that has received funding from the European Research Council (ERC) under the European Union's Horizon 2020 research and innovation programme
Grant Agreement no.~714704.}} 

\author{ 
  Tara Abrishami\thanks{Princeton University, Princeton, NJ 08544} \and
  Maria Chudnovsky\thanks{Princeton University, Princeton, NJ 08544} \and
 Marcin Pilipczuk\thanks{Institute of Informatics, University of Warsaw, Banacha 2, 02-097 Warsaw, Poland} \and
 Pawe\l{} Rz\k{a}\.{z}ewski\thanks{Faculty of Mathematics and Information Science, Warsaw University of Technology, Poland, and 
  Institute of Informatics, University of Warsaw, Banacha 2, 02-097 Warsaw, Poland} \and
  Paul Seymour\thanks{Princeton University, Princeton, NJ 08544}}

\date{}

\begin{document}

\begin{titlepage}
\def\thepage{}
\thispagestyle{empty}
\maketitle

\begin{textblock}{20}(0, 13.0)
\includegraphics[width=40px]{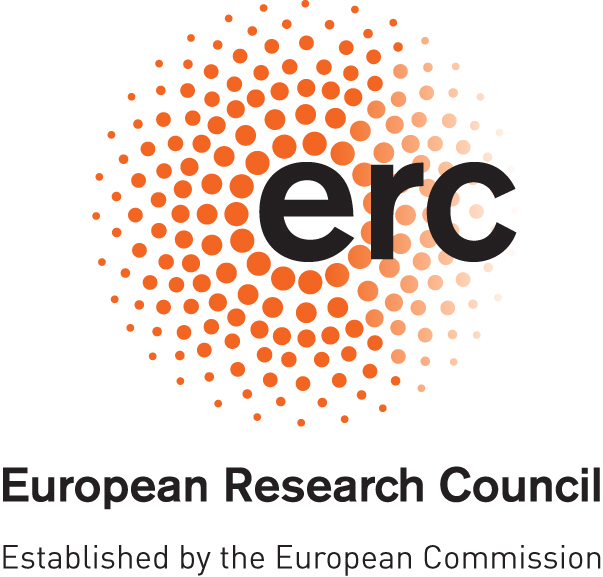}%
\end{textblock}
\begin{textblock}{20}(-0.25, 13.4)
\includegraphics[width=60px]{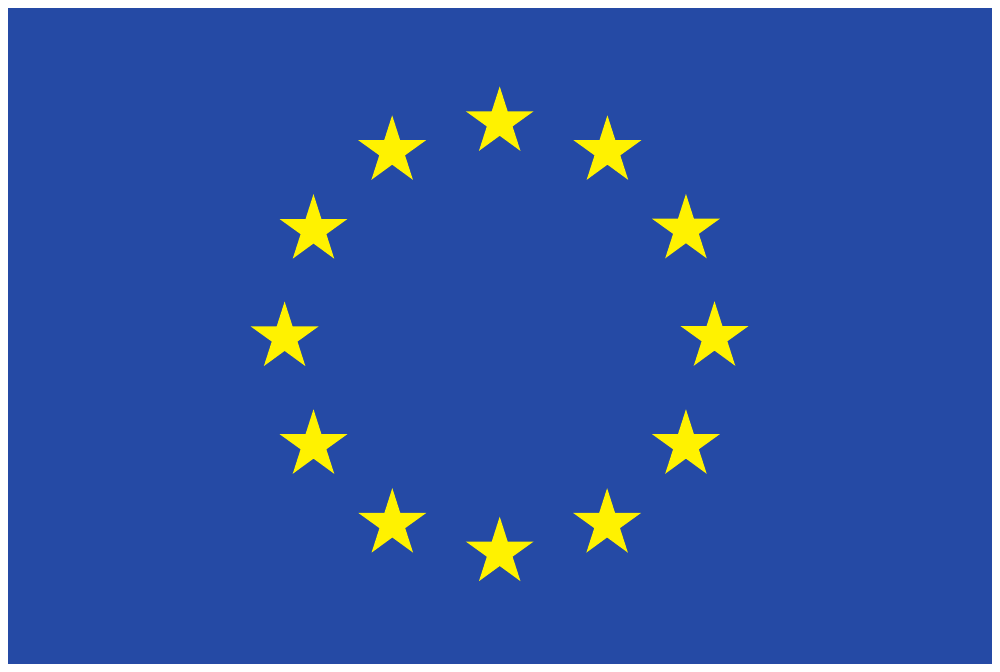}%
\end{textblock}

\begin{abstract}
	A {\em hole} in a graph is an induced cycle of length at least $4$.
	A hole is {\em long} if its length is at least $5$.
  By $P_t$ we denote a path on $t$ vertices. 
	In this paper we give polynomial-time algorithms for the following problems:
  \begin{itemize}
   \item the \textsc{Maximum Weight Independent Set} problem in	long-hole-free graphs, and 
   \item the \textsc{Feedback Vertex Set} problem in $P_5$-free graphs.
  \end{itemize}
  Each of the above results resolves a corresponding long-standing open problem. 

An \emph{extended $C_5$} is a five-vertex hole with an additional vertex adjacent to one or two consecutive vertices of the hole. 
Let $\ourclass$ be the class of graphs excluding an extended $C_5$ and holes of length at least $6$ as induced subgraphs;
$\ourclass$ contains all long-hole-free graphs and all $P_5$-free graphs. We show that, given an $n$-vertex graph $G \in \ourclass$
with vertex weights and an integer $k$, one can in time $n^{\Oh(k)}$ find a maximum-weight induced subgraph of $G$ of treewidth less than $k$.
This implies both aforementioned results.

To achieve this goal, we extend the framework of potential maximal cliques (PMCs) to \emph{containers}. 
Developed by Bouchitt\'{e} and Todinca [SIAM J. Comput. 2001] and extended by Fomin, Todinca, and Villanger [SIAM J. Comput. 2015],
this framework allows to solve high variety of tasks, including finding a maximum-weight induced subgraph of treewidth less than $k$ for fixed $k$,
     in time polynomial in the size of the graph and the number of potential maximal cliques. 
Further developments, tailored to solve the \textsc{Maximum Weight Independent Set} problem within this framework (e.g., for $P_5$-free [SODA 2014] or $P_6$-free graphs [SODA 2019]),
enumerate only a specifically chosen subset of all PMCs of a graph.
In all aforementioned works, the final step is an involved dynamic programming algorithm  whose state space is based on the considered list of PMCs.

Here we modify the dynamic programming algorithm and show that it
is sufficient to consider only a \emph{container} for each potential maximal clique: a
superset of the  maximal clique that intersects the sought solution only in the vertices of the potential maximal clique.
This strengthening of the framework not only allows us to obtain our main result, but also leads to significant simplifications of reasonings in previous papers.
\end{abstract}
\end{titlepage}

\section{Introduction}\label{sec:introduction}
An \emph{independent set} (or \emph{stable set}) in a simple graph $G$ is a set $I \subseteq V(G)$ such that no edge in $E(G)$ has both endpoints in $I$. Given a graph $G$ with non-negative vertex weights, the \textsc{Maximum Weight Independent Set} problem (\textsc{MWIS}) asks for an independent set of $G$ with the greatest total weight.
The \textsc{MWIS} problem is NP-hard in general~\cite{Karp72}.
Over the last several decades researchers have been trying to understand  what restrictions on the input graph
allow efficient algorithms for \textsc{MWIS}.

Given a graph $G$, a {\em hole} in $G$  is an induced cycle of length at least
four, and an {\em antihole} in $G$ is an induced subgraph which is the
complement of an odd cycle of length at least four.%
\footnote{Sometimes a hole is defined to have length at least five, that is, a cycle of length $4$ is not a hole. Since we use the notion of chordal graphs in this work (which are exactly hole-free graphs by our definition), we prefer to treat a four-vertex cycle as a hole and call all other holes \emph{long}.}
 A hole (or antihole) is {\em long} if it has at least five vertices,
{\em even} if it has an even number of vertices, and {\em odd} if it has an
odd number of vertices.
Probably the best known result concerning an efficient algorithm for MWIS is the polynomial-time algorithm for \textsc{MWIS} in perfect graphs due to Gr\"{o}tschel, Lov\'{a}sz, and Schrijver~\cite{GLS}. 
Recall that, by the Strong Perfect Graph Theorem~\cite{CRST}, 
a graph $G$ is perfect if and only if $G$ contains no odd holes and no odd
antiholes.
However, the algorithm of Gr\"{o}tschel, Lov\'{a}sz, and Schrijver~\cite{GLS} relies on the
ellipsoid method. Designing a \emph{combinatorial} polynomial-time
algorithm for \textsc{MWIS} in perfect graphs remains an important open problem. 
Furthermore, the question of the existence of a combinatorial polynomial-time algorithm
for \textsc{Maximum Weight Clique} in perfect graphs without long antiholes was open and received a considerable amount of attention.
Note that, in the complement of the input graph, this task is equivalent to \textsc{MWIS} in perfect graphs with no  long
holes, i.e., graphs with no long holes and no odd antiholes.

Meanwhile, it turned out that our toolbox for proving NP-hardness of \textsc{MWIS} leaves some interesting graph classes where \textsc{MWIS} can be tractable. 
Following the discussion in the previous paragraph, no NP-hardness result is known
for \textsc{MWIS} in {\em long-hole-free graphs}, that is, graphs with
no long holes.
The question of the existence of an efficient algorithm for \textsc{MWIS} in
this graph
class remained a long-standing open problem with a number of tractability results in subclasses~\cite{BasavarajuCK12,BrandstadtG12,BrandstadtGM12,BrandstadtLM10,BrandstadtM15,BerryBGM15}. 
Here we answer this question in the affirmative.

\begin{theorem}
\label{thm:long_hole_poly_time}
The \textsc{Maximum Weight Independent Set} problem in long-hole-free graphs can be solved in polynomial time.\end{theorem}

Similarly, no NP-hardness result for \textsc{MWIS} is  known for $P_t$-free graphs for any $t$,
where $P_t$ is the path on $t$ vertices.
Since $P_4$-free graphs have bounded cliquewidth, many computational problems,
including \textsc{MWIS}, can be solved in $P_4$-free graphs in linear time. 
Only recently, polynomial-time algorithms for \textsc{MWIS} in $P_5$-free~\cite{LokshtanovVV14}
and $P_6$-free graphs~\cite{GrzesikKPP19} were developed; the case of $P_7$-free graphs remains open.

\medskip

Given a graph $G$, the \textsc{Feedback Vertex Set} problem (\textsc{FVS})  asks for a minimum-sized set $X \subseteq V(G)$
such that $G-X$ is a forest. Equivalently, we can ask for a maximum-sized set $Y \subseteq V(G)$
that induces a forest in $G$; the latter formulation is sometimes called
\textsc{Maximum Induced Forest}. 
The problem is one of the classic NP-hard optimization problems, with its directed version
on the Karp's list of 21 NP-hard problems~\cite{Karp72}. 
Similarly as for \textsc{MWIS}, \textsc{FVS} is polynomial-time solvable in $P_4$-free graphs
due to their simple nature, while no NP-hardness nor polynomial-time tractability result is known in $P_t$-free graphs for any $t \geq 5$.
Thus, the complexity of \textsc{FVS} in $P_5$-free graphs remained open
with~\cite{BrandstadtK85,ChiarelliHJMP18,BonamyDFJP19,DFJPPR19} among partial 
results.
In this work, we show tractability of \textsc{FVS} in $P_5$-free graphs.

\begin{theorem}
\label{thm:fvs_poly_time}
The \textsc{Feedback Vertex Set} problem in $P_5$-free graphs can be solved in polynomial time.\end{theorem}

Both Theorem~\ref{thm:long_hole_poly_time} and Theorem~\ref{thm:fvs_poly_time} are straightforward corollaries of the following more general result.
A graph $H$ is an {\em extended} $C_5$ if $H$ is obtained from a five-vertex hole by
adding a simplicial vertex, i.e., a vertex adjacent to one or two consecutive vertices of the cycle.
Let $\ourclass$ be the family of graphs with no hole of length at least $6$ and no extended
$C_5$ as an induced subgraph. 
We prove the following.
\begin{theorem}\label{thm:main}
Given an $n$-vertex graph $G \in \ourclass$ with vertex weights $\weight:V(G) \to \N$
and an integer $k$, one can in time $n^{\Oh(k)}$ find a maximum-weight induced subgraph
of $G$ of treewidth less than $k$.
\end{theorem}
In Theorems~\ref{thm:long_hole_poly_time} and~\ref{thm:main} and in the remainder of the paper, we assume that addition and comparison of weights of subsets
of vertices of $G$ can be done in constant time. 
The definitions of treewidth and tree decompositions can be found in Section~\ref{sec:preliminaries}.

Since $\ourclass$ contains all $P_5$-free graphs and all long-hole-free graphs,
while a set $Y \subseteq V(G)$ is independent if and only if $Y$ induces a graph of treewidth
less than $1$ and $Y$ induces a forest if and only if $Y$ induces a graph of treewidth less than $2$, 
Theorem~\ref{thm:main} directly implies Theorem~\ref{thm:long_hole_poly_time} and
Theorem~\ref{thm:fvs_poly_time}. 
It also generalizes the result of Lokshtanov, Villanger, and Vatshelle~\cite{LokshtanovVV14}
on tractability of \textsc{MWIS} in $P_5$-free graphs.

\paragraph{The framework of potential maximal cliques.}
A cornerstone technique for solving the \textsc{MWIS} problem in various graph classes was
introduced by Bouchitt\'{e} and Todinca~\cite{BouchitteT01,BouchitteT02}.
To explain it in more detail, we need some definitions (see also Section~\ref{sec:preliminaries} for the notation).

A graph is \emph{chordal} if it contains no holes. 
Equivalently, a graph is chordal if it admits a tree decomposition where
every bag is a maximal clique.

Let $G$ be a graph. A set $S \subseteq V(G)$ is a \emph{minimal separator}
if there are two distinct connected components $A,B$ of $G-S$ with $N(A) = N(B) = S$.
A set $\mathcal{E} \subseteq \binom{V(G)}{2} \setminus E(G)$ 
is a \emph{chordal completion}
or \emph{fill-in} of $G$ if $G+\mathcal{E} := (V(G), E(G) \cup \mathcal{E})$ is chordal; a chordal completion is \emph{minimal}
if it is inclusion-wise minimal. 
A set $\pmc \subseteq V(G)$ is a \emph{potential maximal clique} (PMC) if
there exists a minimal chordal completion $\mathcal{E}$ such that $\pmc$ is a maximal clique
in $G+\mathcal{E}$.
A graph class $\mathcal{G}$ has  \emph{a polynomial number of minimal separators (PMCs)}
if there exists a constant $c$  such that every $G \in \mathcal{G}$ has at most
$(|V(G)|)^c$ minimal separators (PMCs, respectively). 

The core of the contributions of Bouchitt\'{e} and Todinca~\cite{BouchitteT01,BouchitteT02}
can be summarized as follows:
\begin{enumerate}
\item A graph class has a polynomial number of minimal separators if and only if 
it has a polynomial number of PMCs.
\item All minimal separators and all PMCs of a graph can be enumerated in time polynomial
in the input and output. 
\item Given a graph $G$ and a list of all PMCs of $G$, one can solve \textsc{MWIS} in $G$
in time polynomial in $|V(G)|$ and the size of the list.
The algorithm is an involved dynamic programming
algorithm whose state space is based on the list of PMCs of $G$.
\end{enumerate}
Consequently, \textsc{MWIS} is polynomial-time solvable in any class of graphs that has a polynomial
number of PMCs or minimal separators. 
This result generalizes a number of earlier tractability results for specific graph classes.

The framework of Bouchitt\'{e} and Todinca has been generalized by Fomin and Villanger~\cite{FominV10}
and Fomin, Todinca, and Villanger~\cite{FominTV15} to other problems than just \textsc{MWIS},
including the problem of finding a maximum-weight induced subgraph of treewidth less than
$k$ for constant $k$ and satisfying some fixed property expressible in counting monadic
second order logic (CMSO). Note that this general problem includes \textsc{Feedback Vertex Set}.

However, the above technique has limitations. Consider the following example.
A \emph{$p$-prism} is a graph consisting of two cliques of size $p$ and a matching of their vertices.
More precisely, a $p$-prism has vertex set $\{a_1, \hdots, a_p, b_1, \hdots, b_p\}$, and its set of edges consists of the pairs of the 
following form:  $a_ia_j$ and $b_ib_j$ for $1 \leq i < j \leq p$, and $a_ib_i$ for $1 \leq i \leq p$.
It is easy to see that a $p$-prism has $2^p-2$ minimal separators and
$p2^{p-1}$ PMCs, while being $P_5$-free and long-hole-free. 
Thus, the framework of Bouchitt\'{e} and Todinca per se cannot provide a
polynomial-time algorithm for \textsc{MWIS} in long-hole-free graphs or $P_t$-free
graphs for any $t \geq 5$. 
In~\cite{long-hole-free-subexp}, it is proven that in long-hole-free graphs
$p$-prisms are the only obstacles to a polynomial number of PMCs and minimal separators: 
an $n$-vertex long-hole-free graph without a $p$-prism as an induced subgraph
has $n^{p + \Oh(1)}$ minimal separators.

The complexity of \textsc{MWIS} in $P_5$-free graphs was a long-standing open problem
until 2014, when Lokshtanov, Vatshelle, and Villanger~\cite{LokshtanovVV14}
presented an algorithm based on an ingenious modification of the framework of Bouchitt\'{e}
and Todinca. 
The main engine of their approach is encapsulated in the following statement:
\begin{theorem}[\cite{LokshtanovVV14}]
\label{thm:PMC_nonexhaustive_list}
Given a graph $G$ and a list $\Pi$ of potential maximal cliques of $G$, one can compute in time $O(|\Pi|n^5m)$ the maximum weight independent set $I$, such that there exists a minimal chordal completion $\mathcal{E}$ of $G$ such that every maximal clique $\pmc$ of $G + \mathcal{E}$ is on the list $\Pi$ and satisfies $|\pmc \cap I| \leq 1$.  
\end{theorem}
That is, one no longer requires to list \emph{all} potential maximal cliques of the input graph.
Instead, it is sufficient to find a list $\mathcal{S}$ of polynomial size with the following property:
For the sought solution $I$ there exists a minimal
chordal completion $\mathcal{E}$, such that all maximal cliques of $G+\mathcal{E}$ are in $\mathcal{S}$
and every maximal clique of $G+\mathcal{E}$ intersects $I$ in at most one vertex. 
Based on this modified approach, Grzesik, Klimo\v{s}ov\'{a}, Pilipczuk, and Pilipczuk
presented a polynomial-time algorithm for \textsc{MWIS} in $P_6$-free graphs~\cite{GrzesikKPP19}.

The PMC enumeration algorithm for $P_5$-free graphs of~\cite{LokshtanovVV14} enumerates
PMCs in three steps. 
Let $I$ be the sought solution (an independent set of maximum weight). 
Initially, the algorithm observes that there always exists a minimal chordal completion $\mathcal{E}$
such that no edge of $\mathcal{E}$ is incident with $I$, as completing $V(G) \setminus I$ into a clique 
turns $G$ into a split graph (in particular, a chordal graph). 
Thus, we can restrict to $\mathcal{E}$ being \emph{$I$-safe},
that is, not containing an edge incident with $I$. 
Then, immediately for every maximal clique $\pmc$ of $G+\mathcal{E}$ we have that $|\pmc \cap I| \leq 1$. 
In the first phase of the enumeration, an argument independent of the graph class
handles maximal cliques $\pmc$ of $G+\mathcal{E}$ with $|\pmc \cap I| = 1$. 
The second phase of the enumeration considers maximal cliques $\pmc$ 
that are disjoint from $I$, but contained in the union of the neighborhoods of two elements of $I$.
The third phase of the enumeration handles the remaining maximal cliques. 
As shown in~\cite{LokshtanovVV14}, in $P_5$-free graphs
there is only a polynomial number of PMCs of the third
type (for all choices of a solution $I$ and an $I$-safe minimal chordal completion $\mathcal{E}$)
and they can be enumerated in polynomial time. 
The example of a $p$-prism shows that there can be exponentially many PMCs of the second type.
In~\cite{LokshtanovVV14}, the selection of the PMCs of the second type to list is handled by an insightful argument
specific to $P_5$-free graphs that stops to work in $P_6$-free graphs.
Partially due to this, the work for $P_6$-free graphs~\cite{GrzesikKPP19} is substantially more involved and elaborate.

\paragraph{Our technical contribution.}
In this work, we generalize the framework to \emph{containers} of PMCs.
For an induced subgraph $F$ of $G$,
an \emph{$F$-container}
for a set $\pmc \subseteq V(G)$ is a set $A \subseteq V(G)$ 
such that $\pmc \subseteq A$ and $A \cap V(F) = \pmc \cap V(F)$.
A roughly similar notion of a container first appeared in~\cite{BMS,ST}. 
In Section~\ref{sec:DP} we prove the following:

\begin{restatable}{theorem}{DPtheorem}
\label{thm:DP_intro}
Assume we are given a graph $G$ with weight function $\weight : V(G) \to \N$, 
a family $\coverfam$ of subsets of $V(G)$, and a positive integer $k$ with the following promise:
\begin{displayquote}
For every induced subgraph $F$ of $G$ of treewidth less than $k$
and every potential maximal clique $\Omega$ of $G$
if $|V(F) \cap \Omega| \leq k$ then $\coverfam$ contains an
$F$-container for $\Omega$.
\end{displayquote}
\noindent Then, one can in time $|\coverfam|^2 |V(G)|^{\Oh(k)}$
find a maximum-weight induced subgraph of $(G,\weight)$ of treewidth less than $k$.
\end{restatable}
Going back to the outlined algorithm for \textsc{MWIS} in $P_5$-free graphs of~\cite{LokshtanovVV14},
observe that the following family:
$$\mathcal{F}(G) := \{N[X] \setminus X'~|~X \subseteq V(G) \wedge |X| \leq 2 \wedge X' \subseteq X\}$$
is of size $\Oh(|V(G)|^2)$ and 
contains an $I$-container for every independent set $I$ and 
PMC of the first or second type. 
Thus, with Theorem~\ref{thm:DP_intro} in hand, the algorithm of~\cite{LokshtanovVV14}
can be simplified to only its third phase.
That is, the PMCs of the first and second type are handled  by arguments independent of the studied graph class.

We show how to compute a family $\coverfam$ suitable for Theorem~\ref{thm:DP_intro} for the class $\ourclass$.


\begin{restatable}{theorem}{pmccontainers}
\label{thm:PMC_containers}
Given an $n$-vertex graph $G \in \ourclass$ and an integer $k$, 
one can in $n^{\Oh(k)}$ time compute a family $\mathcal{X}$ of size $\Oh(n^{8k+60})$
such that for every $k$-colorable induced subgraph $F$ of $G$
and every potential maximal clique $\Omega$ of $G$
there exists $\mathcal{S} \in \mathcal{X}$ such that $\mathcal{S}$ is an
$F$-container for $\Omega$. 
\end{restatable}

Due to the notion of containers,
the enumeration algorithm and our reasoning in Theorem~\ref{thm:PMC_containers}
is arguably simpler and shorter
than its counterpart for $P_6$-free graphs~\cite{GrzesikKPP19}.

Theorem~\ref{thm:main} follows by pipelining Theorem~\ref{thm:PMC_containers}
and Theorem~\ref{thm:DP_intro} and observing that a graph of treewidth less than $k$
is always $k$-colorable.

\paragraph{Organization.}
In Section~\ref{sec:preliminaries}, we define minimal separators and potential maximal cliques and review their properties. 
In Section~\ref{sec:containers_separators}, we consider containers for minimal separators,
   which are later used in Section~\ref{sec:containers_PMCs}
   to prove Theorem~\ref{thm:PMC_containers}.
In Section~\ref{sec:DP_algorithm}, we prove Theorem \ref{thm:DP_intro}.
Section~\ref{sec:conclusion} concludes the paper and includes a discussion on possible extensions of Theorem~\ref{thm:DP_intro}.

\section{Preliminaries}\label{sec:preliminaries}
Let $G$ be a graph with vertex set $V(G)$ and edge set $E(G)$. Let $X \subseteq V(G)$. We denote by $G[X]$ the subgraph of $G$ induced by $X$, and by $G - X$ the subgraph induced by $V(G) \setminus X$. The set of connected components of $G - X$ (as a family of vertex sets) is given by $\cc(G - X)$. The \emph{open neighborhood} of $X$ in $G$, denoted $N_G(X)$, is the set of vertices in $V(G) \setminus X$
with a neighbor in $X$. The \emph{closed neighborhood} of $X$ in $G$, denoted $N_G[X]$, is given by $N_G[X] = N_G(X) \cup X$. We write $N(X)$ and $N[X]$ to mean the open and closed neighborhoods of $X$ in $G$ when $G$ is clear from context. If $Y \subseteq V(G)$, we say that $X$ is \emph{complete} to $Y$ if for every $x \in X$ and $y \in Y$ it holds that $xy \in E(G)$. We say that $X$ is \emph{anticomplete} to $Y$ if for every $x \in X$ and $y \in Y$ it holds that $xy \not \in E(G)$. A \emph{path} is a graph $G$ with vertex set $p_1 \hdots p_n$ such that $p_ip_{i+1} \in E(G)$ for $1 \leq i < n$. The \emph{length} of a path is its number of edges.
A \emph{path from $a$ to $b$ through $X$} is a path with endpoints $a$ and $b$ and interior in $X$. If $a$ and $b$ are adjacent, the path from $a$ to $b$ through $X$ is the edge $ab$. 

A \emph{$k$-coloring} of a graph $G$ is a partition of $V(G)$ into $k$ independent sets.
A graph $G$ is \emph{$k$-colorable} if it admits a $k$-coloring.

A \emph{tree decomposition} $(T, \beta)$ of a graph $G$ is a tree $T$ and a function $\beta : V(T) \to 2^{V(G)}$ such that the following properties hold: (1) for every $uv \in E(G)$, there exists $t \in V(T)$ such that $u, v \in \beta(t)$, and (2) for every $v \in V(G)$, the set $\{t \in V(T) : v \in \beta(t)\}$ induces a nonempty connected subgraph of $T$. The sets $\beta(t)$ for $t \in V(T)$ are called the \emph{bags} of $(T,\beta)$. The \emph{width} of the decomposition $(T,\beta)$
is $\max_{t \in V(T)} |\beta(t)| - 1$ and the \emph{treewidth} of a graph is the minimum
possible width of its decomposition.

Let $X \subseteq V(G)$. The set $X$ is a \emph{minimal separator} if there exist $u, v \in V(G)$ such that $u$ and $v$ are in different connected components of $G - X$, and $u$ and $v$ are in the same connected component of $G - Y$ for every $Y \subsetneq X$. The vertices $u$ and $v$ are said to be \emph{separated by $X$}. A component $D \in \cc(G - X)$ is a \emph{full component} for $X$ if $N(D) = X$. A set $X \subseteq V(G)$ is a minimal separator if and only if there are at least two full components for $X$. Two vertices $u, v \in V(G)$ are separated by a minimal separator $X$ if and only if $u$ and $v$ are in different full components for $X$. 

A \emph{potential maximal clique} (PMC) of a graph $G$ is a set $\Omega \subseteq V(G)$ such that $\Omega$ is a maximal clique of $G + F$ for some minimal chordal completion $F$ of $G$. The following result characterizes PMCs: 

\begin{theorem}[\cite{BouchitteT01}]
\label{thm:PMC_characterization}
A set $\Omega \subseteq V(G)$ is a PMC of $G$ if and only if: 
\begin{enumerate}
    \item for every distinct $x, y \in \Omega$ with $xy \not \in E(G)$, there exists $D \in \cc(G - \Omega)$ such that $x, y \in N(D)$. We say that $D$ \emph{covers} the non-edge $xy$. 
    \item for every $D \in \cc(G - \Omega)$ it holds that $N(D) \subsetneq \Omega$
\end{enumerate}
\end{theorem}
Theorem \ref{thm:PMC_characterization} gives an algorithm to test whether a set $\Omega \subseteq V(G)$ is a PMC of $G$ in time $O(mn)$. 
We also have the following result relating PMCs and minimal separators: 

\begin{proposition}[\cite{BouchitteT01}]
\label{prop:PMC_adhesions_are_seps}
Let $\Omega \subseteq V(G)$ be a PMC of $G$. Then, for every $D \in \cc(G - \Omega)$, the set $N(D)$ is a minimal separator of $G$. 
\end{proposition}

%

\section{Containers for minimal separators}\label{sec:containers_separators}
Let $G$ be a graph in $\ourclass$ and let $n := |V(G)|$. 
Fix an integer $k \geq 0$. 
The goal of this section is to construct a family $\mathcal{F} \subseteq V(G)$ of size $n^{\Oh(k)}$,
such that for every $k$-colorable induced subgraph $F$ of $G$ and for every minimal
separator $S$ of $G$, an $F$-container for $S$ belongs to $\mathcal{F}$.

We call minimal separators $S$ such that $S \subseteq N(v)$ for some $v \in V(G)$ \emph{primitive separators}. The following result deals with primitive separators. 
 
\begin{theorem}
\label{thm:seps_in_nbrhd_single_vtx}
Given an  $n$-vertex graph $G$, one can in polynomial time construct a family $\mathcal{F}_0$ of size at most $n^2$ such that all primitive separators belong to $\mathcal{F}_0$. 
\end{theorem}
\begin{proof}
Let 
$$\mathcal{F}_0 := \bigcup_{v \in V(G)} \{N(C) \mid C \in \cc(G - N[v])\}.$$
Suppose $S$ is a minimal primitive separator of $G$ and $S \subseteq N(v)$ for some $v \in V(G)$. Note that $v \notin S$.
Let $D \in \cc(G - S)$ be a full component for $S$ with $v \not \in D$. Then, $D \in \cc(G - N[v])$ and $S = N(D)$, so $S \in \mathcal{F}_0$. 
\end{proof}

In the rest of this section we focus on separators that are not primitive. 
Let $G \in \mathcal{S}$ and $S$ be a minimal separator of $G$.

\begin{figure}[tb]
\begin{center}
\includegraphics{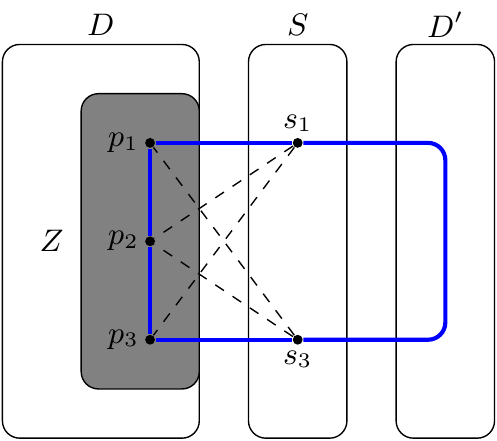}
\caption{Proof of Lemma~\ref{lem:Z_is_a_clique}.}\label{fig:31}
\end{center}
\end{figure}

\begin{lemma}
\label{lem:Z_is_a_clique}
Let $D$ be a full component for $S$. Let $Z \subseteq D$ be a minimal connected subset of $D$ such that $N(Z) = S$. Then, $Z$ is a clique. 
\end{lemma}
\begin{proof}
See  Figure~\ref{fig:31} for an illustration.
Suppose that $Z$ is not a clique.
Let $P = p_1-\hdots-p_t$ be an induced path in $Z$ of maximum possible length.
Then $t \geq 3$.
The maximality of $P$ implies that the sets $Z \setminus \{p_1\}$ and $Z \setminus \{p_t\}$ induce connected subgraphs. The minimality of $Z$ implies that there exist $s_1$ and $s_t$ in $S$ such that $s_1 \in S \cap (N(p_1) \setminus N(Z \setminus \{p_1\}))$ and $s_t \in S \cap (N(p_t) \setminus N(Z \setminus \{p_t\}))$. Let $Q$ be a shortest path from $s_1$ to $s_t$ through a full component $D' \neq D$ for $S$. Since $G \in \ourclass$, and thus $s_1 - p_1 - P - p_t - s_t - Q - s_1$ is not a hole of length at least $6$, we conclude that $t=3$ and $s_1$ is adjacent to $s_t$.
Let $d' \in D'$ be a neighbor of $s_1$. Now $G[s_1,s_t,p_1,p_2,p_3,d']$ is an
extended $C_5$, a contradiction. This proves that $Z$ is a clique.
\end{proof}


\begin{lemma}
\label{lem:f(Z)_is_a_clique}
 Let $D$ be a full component of $S$ and let $Z$ be as in Lemma \ref{lem:Z_is_a_clique}. Then, for every $z \in Z$ there exists $f(z) \in S$ such that $z$ is the unique neighbor of $f(z)$ in $Z$.
 \end{lemma}
\begin{proof}
By Lemma \ref{lem:Z_is_a_clique}, $Z$ is a clique, so $Z \setminus \{z\}$ is connected for all $z \in Z$. By the minimality of $Z$, it follows that for every $z \in Z$ there exists $f(z) \in S$ such that
$z$ is the unique neighbor of $f(z)$ in~$Z$.
\end{proof}

Let $S$ be a minimal separator of $G$ and let $L, R \in \cc(G-S)$ be full components for $S$. Then, there exists $Z \subseteq L$ such that $Z$ is a clique, $N(Z) = S$, and $(N(z) \cap S) \setminus N(Z \setminus \{z\}) \neq \emptyset$  for all $z \in Z$. Similarly, there exists $Z' \subseteq R$  such that $Z'$ is a clique, $N(Z') = S$, and $(N(z) \cap S) \setminus N(Z' \setminus \{z\}) \neq \emptyset$  for all $z \in Z'$. Let $f: Z \cup Z' \to S$ be as defined in Lemma \ref{lem:f(Z)_is_a_clique}, so $f(z) \in (N(Z) \cap S) \setminus N(Z \setminus \{z\})$ for $z \in Z$ and $f(z') \in (N(Z') \cap S) \setminus N(Z' \setminus \{z'\})$ for $z' \in Z'$. 

For every $z \in Z$, recall that $f(z) \in S$ and denote by $g(z) \in R$ an arbitrarily chosen
neighbor of $f(z)$ in $R$. Similarly, for every $z \in Z'$ denote by $g(z) \in L$ an
arbitrarily chosen neighbor of $f(z)$ in~$L$. 

\begin{figure}[tb]
\begin{center}
\includegraphics{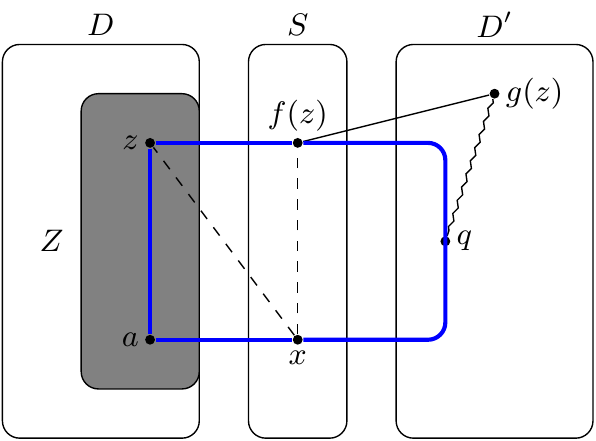}
\caption{Proof of Lemma~\ref{lem:s_hits_all_rungs}. The edge $g(z)q$ may not be present.}\label{fig:34}
\end{center}
\end{figure}
\begin{lemma} 
\label{lem:s_hits_all_rungs}
Let $x \in S$. Then, for every $z \in Z \cup Z'$, it holds that $N(x) \cap \{z, f(z),g(z)\} \neq \emptyset$.
\end{lemma}
\begin{proof}
See Figure~\ref{fig:34} for an illustration.
Let $z \in Z$; the proof for $z \in Z'$ is symmetrical.
The claim is immediate if $x = f(z)$ as $z, g(z) \in N(f(z))$, so assume otherwise.
Suppose that $x$ is anticomplete to $\{z,f(z),g(z)\}$.
Since $N(Z) \supseteq S$, there exists $a \in Z$ such that $xa \in E(G)$. Let $P$ be a shortest path from $x$ to $f(z)$ through $R$.
Then, $z - f(z) - P - x - a - z$ is a hole of length at least six unless $P$ is of length exactly $2$.
If this is the case, then let $q$ be the middle vertex of $P$.
Note that $q \neq g(z)$ as $g(z)$ is nonadjacent to $x$. 
Then $z - f(z) - q - x - a - z$ is a $C_5$. Furthermore, $g(z)$ is adjacent to $f(z)$ and
possibly also $q$. Hence, $G[\{a,z,f(z),g(z),q,x\}]$ is an extended $C_5$, a contradiction.
\end{proof}
 
The set $\mathcal{F}_0$ constructed in Theorem~\ref{thm:seps_in_nbrhd_single_vtx}
contains every primitive separator of $G$.
Therefore, we assume that $S$ is a non-primitive separator, so $|Z|, |Z'| > 1$. Let $a_1, a_2 \in Z$ be distinct. For $i \in \{1,2\}$ let  $b_i = f(a_i)$ and let $r_i = g(a_i)$.
Similarly, let $d_1, d_2 \in Z'$ be distinct, and let $c_i = f(d_i)$ and $l_i = g(d_i)$; see Figure~\ref{fig:W}.
Note that $b_1 \neq b_2$ and $c_1 \neq c_2$ but it may happen that $r_1=r_2$ or $l_1=l_2$.

\begin{figure}[tb]
\begin{center}
\includegraphics{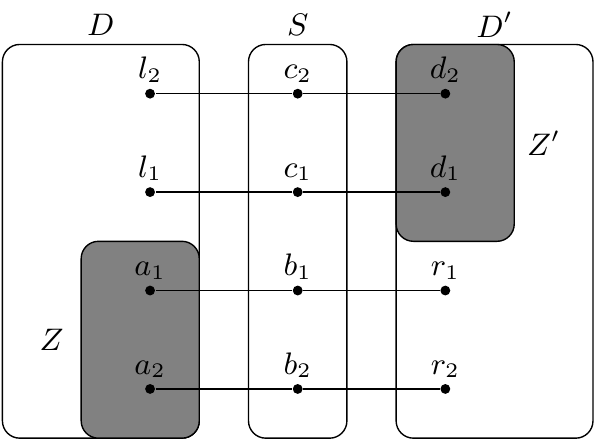}
\caption{Definition of the set $W$.}\label{fig:W}
\end{center}
\end{figure}

Define $W := \{a_1, a_2, b_1, b_2, r_1,r_2,c_1, c_2, d_1, d_2,l_1,l_2\}$. A \emph{profile} is a subset $T \subseteq W$ that meets each of the sets $\{a_1, b_1,r_1\}, \{a_2, b_2,r_2\}, \{c_1, d_1,l_1\}, \{c_2, d_2,l_2\}$. A profile $T$ is \emph{$L$-ambiguous} if $T \subseteq \{a_1, a_2, b_1, b_2, c_1, c_2,l_1,l_2\}$, and \emph{$R$-ambiguous} if $T \subseteq \{b_1, b_2, c_1, c_2, d_1, d_2,r_1,r_2\}$. A profile is \emph{strictly L-ambiguous} if it is $L$-ambiguous and not $R$-ambiguous, and \emph{strictly R-ambiguous} if it is $R$-ambiguous and not $L$-ambiguous.

A few remarks are in place. Lemma~\ref{lem:s_hits_all_rungs} asserts that for every $x \in S$, the set $N(x) \cap W$ is a profile. 
Observe that for every $x \in \{b_1, b_2, c_1, c_2\}$, the profile $N(x) \cap W$ is neither $L$- nor $R$-ambiguous.
Also, note that $\{b_1,b_2,c_1,c_2\}$ is the unique profile that is both $L$- and $R$-ambiguous. 

Let $Z_R$ be the set containing every vertex that is complete to $\{d_1, d_2\}$ and anticomplete to $\{c_1, c_2, l_1, l_2\}$. Similarly, let $Z_L$ be the set containing every vertex that is complete to $\{a_1, a_2\}$ and anticomplete to $\{b_1, b_2, r_1, r_2\}$. We call $Z_R$ and $Z_L$ the \emph{measuring sets} associated to $W$.

\begin{lemma}
\label{lem:s_nbrs_measuring_sets}
If $x \in S$ and $N(x) \cap W$ is an $R$-ambiguous profile, then $x$ has a neighbor in $Z_L \cap L$. Similarly, if $x \in S$ and $N(x) \cap W$ is an $L$-ambiguous profile, then $x$ has a neighbor in $Z_R \cap R$. 
\end{lemma}
\begin{proof}
Let $x \in S$ and let $N(x) \cap W$ be an $R$-ambiguous profile. Because $x \in S$ and $N(Z) \cap S = S$, there exists $y \in Z$ such that $xy \in E(G)$ and $y \neq a_1, a_2$. Then, $y$ is complete to $\{a_1, a_2\}$. Furthermore, $y$ is anticomplete to $\{b_1, b_2\}$ by the definition of $f(\cdot)$
and $y$ is anticomplete to $\{r_1,r_2\}$ as $y \in L$. Hence, $y \in Z_L$. Therefore, $x$ has a neighbor in $Z_L \cap L$. By symmetry, if $N(x) \cap W$ is an $L$-ambiguous profile, then $x$ has a neighbor in $Z_R \cap R$. 
\end{proof}

\begin{figure}[tb]
\begin{center}
\includegraphics{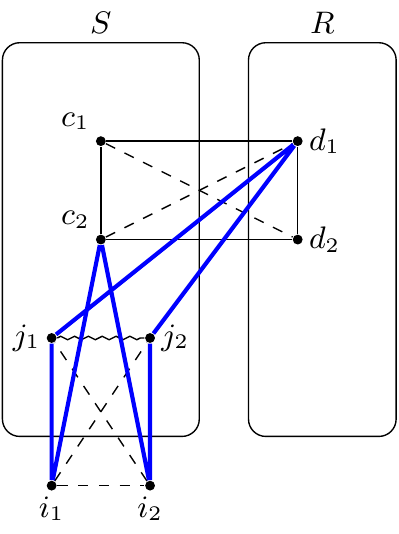}
\caption{Proof of Lemma~\ref{lem:i_nbrs_comparable} in the case 
  when $c_2$ is a common neighbor of $i_1$ and $i_2$;
  otherwise we use also $l_2$ to connect $i_1$ and $i_2$.
The edge $j_1j_2$ may be present.}\label{fig:36}
\end{center}
\end{figure}

\begin{lemma}
\label{lem:i_nbrs_comparable}
Let $I$ be an independent set of $G$. Suppose $i_1, i_2 \in I$ and $N(i_1) \cap W$ and $N(i_2) \cap W$ are both $L$-ambiguous profiles. Then, $N(i_1) \cap Z_R$ and $N(i_2) \cap Z_R$ are comparable in the inclusion order. Similarly, suppose $i_1, i_2 \in I$ and $N(i_1) \cap W$ and $N(i_2) \cap W$ are both $R$-ambiguous profiles. Then, $N(i_1) \cap Z_L$ and $N(i_2) \cap Z_L$ are comparable in the inclusion order.  
\end{lemma}
\begin{proof}
See Figure~\ref{fig:36} for an illustration.
Let $i_1, i_2 \in I$ and suppose $N(i_1) \cap W$ and $N(i_2) \cap W$ are both $L$-ambiguous profiles. Suppose for sake of contradiction that there exist $j_1, j_2 \in Z_R$ such that $i_1j_1, i_2j_2 \in E(G)$ and $i_1j_2, i_2j_1 \not \in E(G)$. Let $P$ be the edge $j_1j_2$ if $j_1j_2 \in E(G)$, and the path $j_1 - d_1 - j_2$ otherwise.
Recall that $N(i_1) \cap W$ and $N(i_2) \cap W$ are profiles, and since $i_1, i_2$ are not adjacent to $d_2$, each of them must have a neighbors in $\{c_2,l_2\}$.
Let $Q$ be a shortest path from $i_1$ to $i_2$ through $\{c_2,l_2\}$; note that $Q$ is
of length $2$ or $3$. Furthermore, by the definition of $Z_R$, the set $\{c_2,l_2\}$ is anticomplete
to $\{j_1,j_2\}$. 
Then $i_1 - j_1 - P - j_2 - i_2 - Q - i_1$ is a hole of length at least six unless
$j_1j_2$ is an edge and $Q$ is of length $2$. 
However, then $i_1 - j_1 - j_2 - i_2 - Q - i_1$, together with $d_1$, induce an extended
$C_5$, a contradiction.

This proves that no such $i_1, i_2, j_1, j_2$ exist, and therefore, $N(i_1) \cap Z_R$ and $N(i_2) \cap Z_R$ are comparable in the inclusion order. By symmetry, if $N(i_1) \cap W$ and $N(i_2) \cap W$ are both $R$-ambiguous profiles, then $N(i_1) \cap Z_L$ and $N(i_2) \cap Z_L$ are comparable in the inclusion order. 
\end{proof}

Now, let $F$ be a $k$-colorable induced subgraph of $G$.
Fix some $k$-coloring of $F$, and let $I_1, \ldots, I_k$ be the partition of $V(F)$ into color classes.
Let $S$ be a minimal separator of $G$.
Recall that our goal is to construct an $F$-container $\sepcont$ of $S$.
For $1 \leq j \leq k$, let $i^j_L \in I_j \setminus (S \cup R)$ be such that $N(i^j_L) \cap W$ is $L$-ambiguous and $N(i^j_L) \cap Z_R$ is inclusion-wise maximal among all vertices of $I_j \setminus (S \cup R)$ with $L$-ambiguous neighbor sets in $W$.
Similarly, let $i^j_R \in I_j \setminus L$ be such that $N(i^j_R) \cap W$ is $R$-ambiguous, and $N(i^j_R) \cap Z_L$ is inclusion-wise maximal among all vertices of $I_j \setminus (S \cup L)$ with $R$-ambiguous neighbor sets in $W$.
We set $i^j_L := \bot$ ($i^j_R := \bot$) if $I_j \setminus (S \cup R) = \emptyset$ ($I_j \setminus (S \cup L) = \emptyset$, respectively)
and in what follows use the convention that $N(\bot) := \emptyset$. 

Let $\sepcont$ be the set containing the following vertices: 
\begin{itemize}
    \item the vertices $b_1, b_2, c_1, c_2$
    \item all vertices $v$ such that $N(v) \cap W$ is an unambiguous profile 
    \item all vertices $v$ such that $N(v) \cap W$ is a strictly $L$-ambiguous profile and $v$ has a neighbor in $Z_R \setminus \bigcup_{j=1}^{k}N(i^j_L)$
    \item all vertices $v$ such that $N(v) \cap W$ is a strictly $R$-ambiguous profile and $v$ has a neighbor in $Z_L \setminus \bigcup_{j=1}^{k} N(i^j_R)$
    \item all vertices $v$ such that $N(v) \cap W$ is $L$-ambiguous and $R$-ambiguous, $v$ has a neighbor in $Z_R \setminus \bigcup_{j=1}^{k}N(i^j_L)$, and $v$ has a neighbor in $Z_L \setminus \bigcup_{j=1}^{k}N(i^j_R)$
\end{itemize}

\begin{lemma}
\label{lem:s_subset_container}
$\sepcont$ is an $F$-container for $S$. 
\end{lemma}
\begin{proof} 

Let $W$, $i^1_R, \ldots, i^{k}_R$, and $i^1_L, \ldots, i^{k}_L$ be as above. First we show that $S \subseteq \sepcont$. Let $s \in S$. By Lemma \ref{lem:s_hits_all_rungs}, $N(s) \cap W$ is a profile for every vertex $s \in S$. If $N(s) \cap W$ is an unambiguous profile, then $s \in \sepcont$. Suppose $N(s) \cap W$ is an $L$-ambiguous profile. By Lemma \ref{lem:s_nbrs_measuring_sets}, there exists $x \in Z_R \cap R$ with $sx \in E(G)$. Since for every $j \in \{1, \ldots, k\}$ we know that $i^j_L \not \in R$ and $i^j_L \not \in S$, it follows that $i^j_Lx \not \in E(G)$ or $i^j_L = \bot$. Therefore, $s$ has a neighbor in
$Z_R \setminus \bigcup_{j=1}^{k}N(i^j_L)$. By symmetry, if $N(s) \cap W$ is an $R$-ambiguous profile, then $s$ has a neighbor in $Z_L \setminus \bigcup_{j=1}^{k}N(i^j_R)$. If $N(s) \cap W$ is strictly $L$-ambiguous, then $s$ has a neighbor in $Z_R \setminus \bigcup_{j=1}^{k}N(i^j_L)$, so $s \in \sepcont$. Similarly, if $N(s) \cap W$ is strictly $R$-ambiguous, then $s$ has a neighbor in $Z_L \setminus \bigcup_{j=1}^{k} N(i^j_R)$, so $s \in \sepcont$. If $N(s) \cap W$ is $L$-ambiguous and $R$-ambiguous, then $s$ has a neighbor in $Z_L \setminus \bigcup_{j=1}^{k} N(i^j_R)$ and a neighbor in $Z_R \setminus \bigcup_{j=1}^{k} N(i^j_L)$, so $s \in \sepcont$.  Therefore, $S \subseteq \sepcont$.

Now we show that $\sepcont \cap V(F) = S \cap V(F)$. Suppose there exists $u \in V(F) \setminus S $;  without loss of generality we may assume that $f \in I_1$.
Recall that for any $v \in V(G)$, if $N(v) \cap W$ is not a profile, then $v \not \in \sepcont$. Therefore suppose that $N(u) \cap W$ is a profile.

We claim that $N(u) \cap W$ is either $L$-ambiguous or $R$-ambiguous.
For contradition, suppose that $N(u) \cap W$ is not ambiguous.
Since $N(u) \cap W$ is not $L$-ambiguous, we observe that $u$ must be adjacent to at least one of $r_1,r_2,d_1,d_2 \in R$.
Similarly, since $N(u) \cap W$ is not $R$-ambiguous, $u$ must be adjacent to at least one of $a_1,a_2,l_1,l_2 \in L$.
Since $u$ has neighbors both in $L$ and in $R$, we conclude that $u \in S$, a contradiction.

If $N(u) \cap W$ is an $L$-ambiguous profile, then Lemma~\ref{lem:i_nbrs_comparable} asserts that $N(u) \cap Z_R$ and $N(i^1_L) \cap Z_R$ are comparable or $i^1_L = \bot$. 
If $i^1_L \neq \bot$ and $N(u) \cap Z_R \subseteq N(i^1_L) \cap Z_R$, then $u \notin \sepcont$ by the definition of $\sepcont$. 
Otherwise, by the choice of $i^1_L$ and since $u \notin S$, we have $u \in R$. This, in turn, implies that $N(u) \cap W$ is an $R$-ambiguous profile. 

By symmetry, we infer that if $N(u) \cap W$ is an $R$-ambiguous profile, then either $u \notin \sepcont$ or $u \in L$. The latter outcome implies that $N(u) \cap W$ is an $L$-ambiguous profile.
Since $u \in L$ and $u \in R$ cannot happen at the same time, we infer that $u \notin \sepcont$. This completes the proof. 
\end{proof}


Now we can finally show an enumeration algorithm for containers of minimal separators.

\begin{theorem}
\label{thm:n^10_sep_containers}
Given an $n$-vertex graph $G \in \ourclass$ and an integer $k$, 
one can in $n^{\Oh(k)}$ time compute a family $\mathcal{F}_1$ of size $\Oh(n^{2k+12})$ such that for every $k$-colorable induced subgraph $F$ of $G$ and every  minimal separator $S$ of $G$
there exists $\sepcont \in \mathcal{F}_1$ such that $\sepcont$ is an
$F$-container for $S$. 
\end{theorem}

\begin{proof}
We first add every separator $S \in \mathcal{F}_0$ to $\mathcal{F}_1$, so $\mathcal{F}_1$ contains all primitive separators of $G$. Next, we enumerate all possible combinations of $W = \{a_1, a_2, b_1, b_2, c_1, c_2, d_1, d_2,l_1,l_2,r_1,r_2\}$, $i^1_R, \ldots, i^{k}_R$, and $i^1_L, \ldots, i^{k}_L$. There are $\Oh(n^{2k+12})$ possibilities for the tuple $(W, i^1_R, \ldots, i^{k}_R, i^1_L, \ldots, i^{k}_L)$. For each tuple $(W, i^1_R, \ldots, i^{k}_R, i^1_L, \ldots, i^{k}_L)$, we add to $\mathcal{F}_1$ the set $\sepcont$ constructed as described above. For every minimal separator $S$ that is not primitive, Lemma \ref{lem:s_subset_container} implies that $ \sepcont$ is an $F$-container for $S$ for the correct choice of $(W,  i^1_R, \ldots, i^{k}_R, i^1_L, \ldots, i^{k}_L)$. Therefore, for every $k$-colorable induced subgraph $F$ of
$G$ and every  minimal separator $S$ of $G$
there exists $\sepcont \in \mathcal{F}_1$ such that $\sepcont$ is an
$F$-container for $S$. 
\end{proof}

In the next section we will need the following strengthening of Theorem \ref{thm:n^10_sep_containers}:

\begin{theorem}
\label{thm:n^11_sep_containers}
Given an $n$-vertex graph $G$ and an integer $k$, 
one can in $n^{\Oh(k)}$ time compute a family $\mathcal{F}_2$ of size $\Oh(n^{2k+13})$ such that for every $k$-colorable induced subgraph $F$ of $G$ and every  minimal separator $S$ of $G$
there exists $\sepcont \in \mathcal{F}_2$ such that $\sepcont$ is an
$F$-container for $S$. 

Furthermore, for every $k$-colorable induced subgraph $F$ of $G$, every 
minimal separator $S$ of $G$ such that $S \notin \mathcal{F}_2$, and 
every two full components $L$ and $R$ of $S$,
there exist $z_\ell,z_r \in S$ with
$N(z_\ell) \cap (V(F) \setminus (S \cup L)) = \emptyset$ and
$N(z_r) \cap (V(F) \setminus (S \cup R)) = \emptyset$.
\end{theorem}
\begin{proof}
Let $\mathcal{F}_2 := \mathcal{F}_1 \cup \{N(D) \mid D \in \cc(G - \sepcont), \ \sepcont \in \mathcal{F}_1\}$. For every $\sepcont \in \mathcal{F}_1$, there are at most $n$ components in $\cc(G-\sepcont)$, so there are $\Oh(n^{2k+13})$ elements in $\mathcal{F}_2$. Let $S \notin \mathcal{F}_2$ be a minimal separator of $G$, let $L$ and $R$ be two full components of $S$, and let $F$ be a $k$-colorable induced subgraph of $G$.

Consider the $F$-container $\sepcont$ for $S$ that is added to $\mathcal{F}_1$ 
for $L$, $R$, a $k$-coloring $F_1,F_2,\ldots,F_k$ of $F$,
and a tuple $(W,i^1_R,\ldots,i^k_R,i^1_L,\ldots,i^k_L)$. 
If $\sepcont \cap L = \emptyset$, then, $L \in \cc(G - \sepcont)$ and $N(L) = S$, so $S \in \mathcal{F}_2$.
So $S \notin \mathcal{F}_2$ implies $\sepcont \cap L \neq \emptyset$ and, symmetrically, $\sepcont \cap R \neq \emptyset$. 

Let $x \in \sepcont \cap L$, and let $y \in \sepcont \cap R$. The vertex $x$ was added to $\sepcont$ because $N(x) \cap W$ is an $L$-ambiguous profile and $x$ has a neighbor in $Z_R \setminus \bigcup_{j=1}^{k}N(i^j_L)$. Let $z_r \in Z_R \setminus \bigcup_{j=1}^{k} N(i^j_L)$ such that $xz_r \in E(G)$.
Recall that all vertices from $Z_R$ are adjacent to both $d_1,d_2 \in R$. 
Since $z_r$ is adjacent to a vertex in $L$ (the vertex $x$) and a vertex in $R$ (e.g., the vertex $d_1$), we conclude that $z_r \in S$.
Therefore, $z_r \in (Z_R \cap S) \backslash \bigcup_{j=1}^{k} N(i^j_L)$.

Because $z_r \not \in \bigcup_{j=1}^{k} N(i^j_L)$ and for every $j \in \{1, \ldots, k \}$, the vertex $i^j_L$ is the vertex in
$F_j \setminus (S \cup R)$ whose neighborhood in $Z_R$ is maximal, we deduce that $z_D$ is anticomplete to $V(F) \setminus (S \cup R)$. The definition and reasoning for $z_\ell$ is symmetrical. This finishes the proof.
\end{proof}

\section{Containers for PMCs}\label{sec:containers_PMCs}
Let again $k$ be a fixed constant and $G \in \ourclass$ be an $n$-vertex graph.
In this section, we describe how to construct a set of containers for the potential maximal cliques of $G$.

The \emph{adhesions} of $\Omega$ are the minimal separators $N(D)$ for $D \in \cc(G - \Omega)$.
We say that $\Omega$ is {\em pure} if all adhesions of $\Omega$ are in $\mathcal{F}_2$, i.e., the family of sets given by Theorem~\ref{thm:n^11_sep_containers}. A PMC that is not pure is called {\em impure}.
  
The following two lemmas are slight strengthenings of results from \cite{long-hole-free-subexp}. 
\begin{lemma}[\cite{long-hole-free-subexp}]
\label{lemma:indset_in_PMC_covered}
Let $G \in \mathcal{C}$ and $\Omega \subseteq V(G)$ be a PMC of $G$, and suppose $J \subseteq \Omega$ is an independent set with $|J| > 1$. Then, there exists $D \in \cc(G - \Omega)$ such that $J \subseteq N(D)$. 
\end{lemma}
\begin{proof}
If $|J| = 2$, the result follows from Theorem \ref{thm:PMC_characterization}, so assume $|J| \geq 3$. Let $D_1 \in \cc(G - \Omega)$ be the component of $G - \Omega$ that maximizes $|N(D_1) \cap J|$, and suppose $J \setminus N(D_1) \neq \emptyset$. Since every nonedge of $J$ is covered by some component, $|N(D_1) \cap J| \geq 2$.
Let $D_2 \in \cc(G - \Omega)$ be the component that maximizes $|J \cap N(D_1) \cap N(D_2)|$ subject to $N(D_2) \cap (J \setminus N(D_1)) \neq \emptyset$.
Since for every $x \in J \cap N(D_1)$ and $y \in J \setminus N(D_1)$ there exists a component covering the nonedge $xy$, such a component $D_2$ exists and $J \cap N(D_1) \cap N(D_2) \neq \emptyset$.
By the choice of $D_2$, there exists $y \in J \cap (N(D_2) \setminus N(D_1))$. By the maximality of $|N(D_1) \cap J|$, there exists $x \in J \cap (N(D_1) \setminus N(D_2))$.
By Theorem~\ref{thm:PMC_characterization}, there exists a component $D_3 \in \cc(G-\Omega)$ covering the nonedge $xy$; note that $D_3 \neq D_1,D_2$. 
By the maximality of $D_2$, as $x \in J \cap ((N(D_1) \cap N(D_3)) \setminus N(D_2))$ and $y \in J \cap ((N(D_2) \cap N(D_3)) \setminus N(D_1))$, there exists $z \in J \cap ((N(D_1) \cap N(D_2)) \setminus N(D_3))$. Let $P_1$ be a shortest path from $x$ to $z$ via $D_1$, let $P_2$ be a shortest path from $z$ to $y$ through $D_2$, and let $P_3$ be a shortest path from $x$ to $y$ through $D_3$.
Then $x - P_1 - z - P_2 - y - P_3 - x$ is a hole of length at least six (see Figure~\ref{fig:41}), a contradiction.
\end{proof}

\begin{figure}[tb]
\begin{center}
\includegraphics{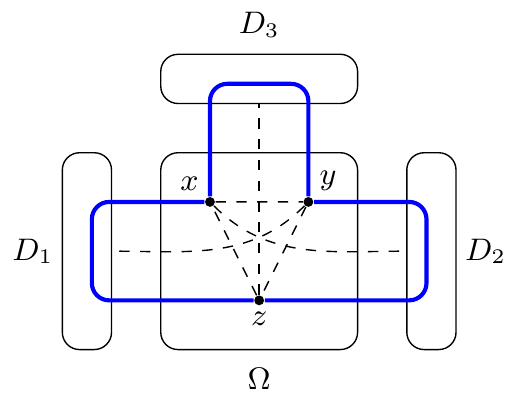}
\caption{Proof of Lemma~\ref{lemma:indset_in_PMC_covered}.}\label{fig:41}
\end{center}
\end{figure}

\begin{figure}[tb]
\begin{center}
\includegraphics{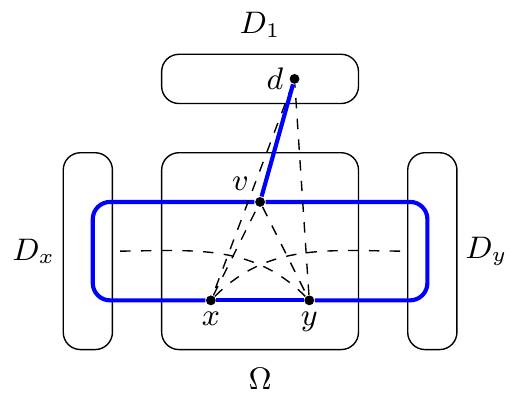}
\caption{Proof of Lemma~\ref{lemma:PMC_minus_nbrhd_covered}.}\label{fig:42}
\end{center}
\end{figure}

\begin{lemma}[\cite{long-hole-free-subexp}]
\label{lemma:PMC_minus_nbrhd_covered}

Let $G \in \mathcal{C}$ , let $\Omega$ be a PMC of $G$, and let $v \in \Omega$ be such that at least one component $D \in \cc(G-\Omega)$
satisfies $v \in N(D)$. Then there exist $D_1,D_2 \in \cc(G - \Omega)$ with $v \in N(D_1) \cap N(D_2)$, and $\Omega \setminus N(v) \subseteq N(D_1) \cup N(D_2)$.
\end{lemma}
\begin{proof}
Let $D_1 \in \cc(G - \Omega)$ such that $v$ has a neighbor in $D_1$.
Let $d \in D_1 \cap N(v)$.
Suppose that there is no $D   \in \cc(G - \Omega)$ with $v \in N(D)$ such that
$\Omega \setminus (N(v) \cup N(D_1)) \subseteq N(D)$.
Let $M \subseteq \Omega \setminus (N(v) \cup N(D_1)) $ be such that $M \cup \{v\} \not \subseteq N(D)$ for all $D \in \cc(G - \Omega)$, and $M' \cup \{v\} \subseteq N(D)$ for some $D \in \cc(G - \Omega)$ for every $M' \subsetneq M$.
If $M$ is an independent set, $M \cup \{v\}$ is also an independent set, so by Lemma \ref{lemma:indset_in_PMC_covered} we conclude that $M \cup \{v\} \subseteq D$ for some $D \in \cc(G - \Omega)$. Therefore, there exist $x, y \in M$ such that $xy \in E(G)$. 

By the definition of $M$ we know that $(M \setminus \{x\}) \cup \{v\} \subseteq D_y$ for some $D_y \in \cc(G-\Omega)$. Similarly, $(M \setminus \{y\}) \cup \{v\} \subseteq D_x$ for some $D_x \in \cc(G - \Omega)$.
The definition of $M$ implies that $x \notin N(D_y)$ and $y \notin N(D_x)$, so in particular $D_x \neq D_y$. Moreover, $D_1 \neq D_x,D_y$, as $x,y \notin N(D_1)$.
Let $P_x$ be a shortest path from $x$ to $v$ through $D_x$ and $P_y$ be a shortest path from $y$ to $v$ through $D_y$. Then,  $v - P_x - x - y - P_y - v$ is a hole of length at least six
or $G[V(P_x) \cup V(P_y) \cup \{v,d\}]$ is an extended $C_5$ (see Figure~\ref{fig:42}), a contradiction. 
\end{proof}

We now construct a set of $F$-containers for impure  PMCs $\Omega$.

\begin{theorem}
\label{thm:containers_PMCs_bad_adhesions}
Given an $n$-vertex graph $G \in \ourclass$ and an integer $k$,
one can in $n^{\Oh(k)}$ time compute a family $\mathcal{X}_1$ of size $\Oh(n^{8k+54})$ such that for every $k$-colorable induced subgraph $F$ of $G$
and every impure PMC $\Omega$
of $G$, some member of $\mathcal{X}_1$ is an $F$-container for $\Omega$. 
\end{theorem}
\begin{proof}
Define
$$\mathcal{X}_1 := \left\{ \left( \bigcup \mathcal{Z} \right) \cup (N(u) \cap N(v)) \mid \mathcal{Z} \subseteq \mathcal{F}_2, |\mathcal{Z}| \leq 4, u, v \in V(G)\right\}.$$
There are $\Oh(n^{2k+13})$ elements in $\mathcal{F}_2$ and $n$ elements in $V(G)$, so $\mathcal{X}_1$ has size $\Oh(n^{8k+54})$.

Suppose $F$ is a $k$-colorable induced subgraph of $G$ and $\Omega$ is an impure PMC of $G$.
Let $S$ be an adhesion of $\Omega$, such that $S \notin \mathcal{F}_2$. Let $L$ be a component of $G-\Omega$
such that $S = N(L)$, and let $R$ be another full component of $S$.
By Theorem~\ref{thm:n^11_sep_containers}, as $S \notin \mathcal{F}_2$, 
there exist $z_\ell, z_r \in  S$ such that
$N(z_\ell) \cap (V(F) \setminus (L \cup S)) = \emptyset$ and
$N(z_r) \cap (V(F) \setminus (R \cup S))=\emptyset$. 
Since $N(L) = S$ and $L \in \cc(G-\Omega)$, each of $z_\ell$ and $z_r$ has a neighbor
in $V(G) \setminus \Omega$.
By Lemma \ref{lemma:PMC_minus_nbrhd_covered}, there exist minimal
separators 
$S_1^{\ell}, S_2^{\ell}, S_1^r,S_2^r$ of $G$, all contained in $\Omega$, such that
$\Omega \setminus N(z_\ell) \subseteq S_1^{\ell} \cup S_2^{\ell}$ and
$\Omega \setminus N(z_r) \subseteq S_1^r \cup S_2^r$.
Then, $\Omega \subseteq S^{\ell}_1 \cup S_2^{\ell} \cup S_1^r \cup S_2^r \cup (N(z_\ell) \cap N(z_r))$. 

For $i=1,2$, pick $\sepcont^{\ell}_i,\sepcont^r_i,  \in \mathcal{F}_2$ such that $\sepcont_i^{\ell}$ is an $F$-container
for $S_i^{\ell}$ and $\sepcont_i^{r}$ is an $F$-container for $S_i^r$.
Consider the set 
$$\widehat{\Omega} := \sepcont_1^{\ell} \cup \sepcont_2^{\ell} \cup \sepcont_1^{r} \cup \sepcont_2^{r} \cup (N(z_\ell) \cap N(z_r)).$$
Clearly, $\Omega \subseteq \widehat{\Omega}$ and $\widehat{\Omega} \in \mathcal{X}_1$.
Since $z_\ell$ is anticomplete to $V(F)  \setminus (S \cup L)$ and $z_r$ is anticomplete to $V(F)  \setminus (S \cup R)$, it follows that $(N(z_\ell) \cap N(z_r)) \cap V(F) \subseteq V(F) \cap S \subseteq \Omega$.
Since $\sepcont^\ell_i$ is an $F$-container for $S^\ell_i$ and $S^\ell_i \subseteq \Omega$, we obtain that $\sepcont^\ell_i \cap V(F) \subseteq \Omega \cap V(F)$ for $i=1,2$;
a symmetric statement holds for $\mathcal{S}^r_i$. 
Hence, $\widehat{\Omega} \cap V(F) \subseteq \Omega \cap V(F)$.
Since $\Omega \subseteq \widehat{\Omega}$, we conclude that $\widehat{\Omega}$ is an $F$-container for $\Omega$. 
\end{proof}

Next, we aim to construct a set $\mathcal{X}_2$ such that every pure PMC $\Omega$ of $G$ belongs to $\mathcal{X}_2$.
To this end, we 
follow a methodology of \emph{survival sequences}, implicit in \cite{BouchitteT01}, and made explicit and formalized in \cite{GrzesikKPP19}. 
We follow the notation of the full version~\cite{GrzesikKPP19arxiv}.
A sequence $\mathcal{S} = (x_1,x_2,\ldots,x_t)$ of distinct vertices of $G$ is a \emph{survival sequence} for a PMC $\pmc$
if for every $0 \leq i \leq t$ the set $\pmc \setminus \{x_1,x_2,\ldots,x_i\}$ is a PMC in the graph $G-\{x_1,x_2,\ldots,x_i\}$. 
We denote $V(\mathcal{S}) = \{x_1,x_2,\ldots,x_t\}$ and we say that $\mathcal{S}$ \emph{ends in} $\pmc \setminus V(\mathcal{S})$, which is a PMC in $G-V(\mathcal{S})$. 
We need the PMC Lifting Lemma from~\cite{GrzesikKPP19,GrzesikKPP19arxiv}.
\begin{lemma}[{PMC Lifting Lemma~\cite[Lemma~22]{GrzesikKPP19arxiv}}]
\label{lem:PMClift}
Let $G$ be a graph and let $\mathcal{S} = (x_1,x_2,\ldots,x_t)$ be a sequence of
distinct vertices of $G$. 
Then for every $\pmc'$ that is a PMC in $G-V(\mathcal{S})$, there exists a unique $\pmc$
that is a PMC in $G$ and $\mathcal{S}$ is a survival sequence for $\pmc$ ending in $\pmc'$.
Moreover, given $G$, $S$, and $\pmc'$, the PMC $\pmc$ can be computed in polynomial time.
\end{lemma}

The next lemma and its proof is the analog of Lemma~25
of~\cite{GrzesikKPP19arxiv}.

\begin{lemma}\label{lem:PMClift-long-hole-free}
Suppose $G \in \mathcal{C}$ and let $n=|V(G)|$.
Given a family $\mathcal{Y} \subseteq 2^{V(G)}$, one can in time $(n \cdot |\mathcal{Y}|)^{\Oh(1)}$ compute a family $\mathcal{X}_\mathrm{rec}(\mathcal{Y}) \subseteq 2^{V(G)}$,
such that $|\mathcal{X}_\mathrm{rec}(\mathcal{Y})| \leq 3n^4|\mathcal{Y}|^4$
and the following property holds: for every PMC $\pmc$ in $G$, if $\cc(G-\pmc) \subseteq \mathcal{Y}$,
    then $\pmc \in \mathcal{X}_{\mathrm{rec}}(\mathcal{Y})$.
\end{lemma}
\begin{proof}
Let $\pmc$ be a PMC in $G$, such that $\cc(G-\pmc) \subseteq \mathcal{Y}$.
Let $x_1,x_2,\ldots,x_n$ be an arbitrary enumeration of $V(G)$
and for $0 \leq i \leq n$, let $X_i := \{x_1,x_2,\ldots,x_i\}$ (where $X_0 := \emptyset$),
$G_i := G-X_i$, and $\pmc_i := \pmc \setminus X_i$.
Let $0 \leq s \leq n$ be the maximum integer such that $\pmc_s$ is a PMC in $G_s$;
since $\pmc = \pmc_0$ is a PMC in $G = G_0$, such an integer exists.

Since $\cc(G-\pmc) \subseteq \mathcal{Y}$, we have $\cc(G_s-\pmc_s) \subseteq \mathcal{Y}_s$,
where 
$\mathcal{Y}_s: = \bigcup_{D \in \mathcal{Y}} \cc(G[D\setminus X_s])$.
Note that $|\mathcal{Y}_s| \leq (n-s)|\mathcal{Y}|$.

If $s=n$, then $(x_1,x_2,\ldots,x_n)$ is a survival sequence for $\pmc$ ending in
$\pmc_n = \emptyset$ in an empty graph $G_n$.
By Lemma~\ref{lem:PMClift}, there is exactly one such PMC $\pmc_\emptyset$ and it can be computed in polynomial time. We define $\mathcal{G}_0 = \{\pmc_\emptyset\}$. 

Assume then $s < n$ and let $v := x_{s+1}$.
Then $\Omega_{s+1} = \Omega_s \setminus \{v\}$
is not a PMC in $G_{s+1} = G_s - \{v\}$ due to the choice of $s$.

First, suppose $v \in \Omega$.
Then, $\cc(G_s - \Omega_s) = \cc(G_{s+1} - \Omega_{s+1})$.
Therefore, for every nonedge $xy$ in $\Omega_{s+1}$, there exists a component $D \in \cc(G_{s+1}-\Omega_{s+1})$ that covers $xy$. It follows that $\Omega_{s+1}$ is not a PMC of $G_{s+1}$ because for some $D \in \cc(G_{s+1}-\Omega_{s+1})$ it holds that $N_{G_{s+1}}(D) = \Omega_{s+1}$. Then, $\Omega_s = N_{G_s}(D) \cup \{v\}$. 
Thus, $\Omega \in \mathcal{G}_1$ where
$\mathcal{G}_1$ is constructed as follows: for every $0 \leq s < n$
and every $D \in \mathcal{Y}_s$, compute $Z := N_{G_s}(D) \cup \{x_{s+1}\}$
and if $Z$ is a PMC in $G_s$, apply the PMC Lifting Lemma to 
the graph $G$, the sequence $(x_1,x_2,\ldots,x_s)$ and the PMC $Z$, and insert the resulting
PMC of $G$ into $\mathcal{G}_1$.
Note that $|\mathcal{G}_1| \leq \sum_{s=0}^{n-1} (n-s)|\mathcal{Y}| = \binom{n+1}{2} |\mathcal{Y}|$.

Now, suppose $v \not \in \Omega$. Then, $\pmc_s = \pmc_{s+1}$ and
$v \in D$ for some $D \in \cc(G_s-\Omega_s)$.
For every $D' \in \cc(G_{s+1} - \Omega_{s+1})$, either $D' \in \cc(G_s - \Omega_s)$ or $D' \subseteq D$, so $N_{G_{s+1}}(D') \subsetneq \Omega_{s+1}$ for all $D' \in \cc(G_{s+1}-\Omega_{s+1})$.
It follows that $\Omega_{s+1}$ is not a PMC in $G_{s+1}$, because some nonedge $xy$ in $\Omega_{s+1}$ is not covered by a component in $\cc(G_{s+1} - \Omega_{s+1})$. Therefore, $D$ is the unique component in $\cc(G_s - \Omega_s)$ covering $xy$. Furthermore, $v \in D$ and $(N(x) \cap N(y)) \setminus \Omega_s \subseteq \{v\}$.
By  Lemma \ref{lemma:PMC_minus_nbrhd_covered}, there exist
$D_1, D_2, D_3,D_4 \in  \cc(G_s - \Omega_s)$ such that
\[\Omega_s = \left(  \left( \bigcup_{1 \leq i \leq 4} N_{G_s}(D_i) \right) \cup \Bigl(N(x) \cap N(y)\Bigr) \right) \setminus \{v\}.\]
Hence, $\Omega \in  \mathcal{G}_2$ where $\mathcal{G}_2$ is constructed as follows:
for every $0 \leq s < n$, for every $D_1,D_2,D_3,D_4 \in \mathcal{Y}_s$,
and for every $x, y \in V(G)$, compute
\[Z :=  \left(  \left( \bigcup_{1 \leq i \leq 4} N_{G_s}(D_i) \right) \cup \Bigl(N(x) \cap N(y)\Bigr) \right) \setminus \{x_{s+1}\},\]
and if $Z$ is a PMC in $G_s$, apply the PMC Lifting Lemma to 
the graph $G$, the sequence $(x_1,x_2,\ldots,x_s)$, and the PMC $Z$, and insert the resulting
PMC of $G$ into $\mathcal{G}_2$.
Note that $|\mathcal{G}_2| \leq \sum_{s=0}^{n-1} \binom{n}{2}(n-s)|\mathcal{Y}|^4 = \binom{n+1}{2}\binom{n}{2} |\mathcal{Y}|^4$.

We output $\mathcal{X}_\mathrm{rec}(\mathcal{Y}) := \mathcal{G}_0 \cup \mathcal{G}_1 \cup \mathcal{G}_2$. By the above estimations, for $n > 1$
the output is of size at most $3n^4|\mathcal{Y}|^4$, while for $n=1$ the output is of size at most $2$.
\end{proof}

 We can now construct a set containing all pure PMCs of $G$. 
\begin{theorem}
\label{thm:PMCs_correct_adhesions}
Given an $n$-vertex graph $G \in \ourclass$ and an integer $k$,
one can in time $n^{\Oh(k)}$ construct a set $\mathcal{X}_2$ of size $\Oh(n^{8k+60})$ such that every pure PMC $\Omega$ of $G$ belongs to $\mathcal{X}_2$. 
\end{theorem}
\begin{proof}
We apply Lemma~\ref{lem:PMClift-long-hole-free}
to $G$ and $\mathcal{Y} := \bigcup_{S \in \mathcal{F}_2} \cc(G-S)$, where
$\mathcal{F}_2$ comes from Theorem~\ref{thm:n^11_sep_containers}.
Since $\mathcal{F}_2 = \Oh(n^{2k+13})$, we obtain that $|\mathcal{Y}| = \Oh(n^{2k+14})$ and the size bound follows.
\end{proof}

Finally, we can combine the results of Theorems~\ref{thm:containers_PMCs_bad_adhesions} and~\ref{thm:PMCs_correct_adhesions}, giving the following. 

\pmccontainers*



\section{Dynamic programming algorithm}\label{sec:DP}\label{sec:DP_algorithm}
The goal of this section is to prove Theorem \ref{thm:DP_intro}.
\DPtheorem*

Contrary to the previous two sections, here it is more convenient use the terms
of tree decompositions instead of chordal completions. 
Our main technical statement is the following theorem: 

\begin{theorem}
  \label{thm:DP}
  Assume we are given a graph $G$ with weight function $\weight : V(G) \to \N$, 
	a family $\coverfam$ of subsets of $V(G)$, and a positive integer $k$ with the following promise:
  \begin{displayquote}
  For every induced subgraph $F$ of $G$ of treewidth less than $k$
  there exists a tree decomposition $(T,\beta)$ of $G$
  such that
  \begin{itemize}
  \item for every $t \in V(T)$, an $F$-container for  $\beta(t)$ belongs to $\coverfam$,
  \item $(T,\beta_F)$ is a tree decomposition of $F$ of width less than $k$,
  where $\beta_F(t) := \beta(t) \cap V(F)$ for every $t \in V(T)$.
  \end{itemize}
  \end{displayquote}
\noindent  Then, 
	one can in time $|\coverfam|^2 |V(G)|^{\Oh(k)}$
	find a maximum-weight induced subgraph of $(G,\weight)$ of treewidth less than $k$.
\end{theorem}

We show how Theorem~\ref{thm:DP_intro} follows from Theorem~\ref{thm:DP} in Section~\ref{ssec:DP2DP}
and prove Theorem~\ref{thm:DP} in Section~\ref{ssec:DP_proof}.

\subsection{Proof of Theorem~\ref{thm:DP_intro}}\label{ssec:DP2DP}

We need the following facts on relations between chordal completions and tree decompositions.
The first one is straightforward. 

\begin{proposition}\label{prop:td2fillin}
Let $G$ be a graph and let $(T,\beta)$ be a tree decomposition of $G$.
Then 
$$\mathcal{E} := \bigcup_{t \in V(T)} \binom{\beta(t)}{2} \setminus E(G)$$
is a chordal completion of $G$. 
Consequently, if $G$ has treewidth less than $k$, then
there exists a minimal chordal completion $\mathcal{E}$ of $G$
such that every clique of $G+\mathcal{E}$ is of size at most $k$.
\end{proposition}

The second one is a well-known characterization of chordal graphs.
\begin{proposition}[see e.g. \cite{LokshtanovVV14}]
\label{prop:chordal_clique_tree}
	A graph $G$ is chordal if and only if there exists a tree decomposition $(T, \beta)$ of $G$ such that every bag is a maximal clique in $G$. If $G$ is chordal, such a tree decomposition is called a \emph{clique tree} of $G$. 
\end{proposition}

The third one has been pivotal to the results of~\cite{FominV10,FominTV15}.
\begin{lemma}[{\cite[Lemma~3.1]{FominV10}, \cite[Lemma~2.9]{FominTV15}}]\label{lem:FTV}
Let $F$ be an induced subgraph of $G$ and let $\mathcal{E}_F$ be a minimal chordal completion
of $F$. Then there exists a minimal chordal completion $\mathcal{E}_G$ of $G$
such that for every clique $\Omega$ of $G+\mathcal{E}_G$, the intersection $\Omega \cap V(F)$ 
is either empty or is a clique of $F+\mathcal{E}_F$. 
\end{lemma}

Consider the input tuple $(G,\weight,\coverfam,k)$ as in Theorem~\ref{thm:DP_intro}.
We claim that we can pass the same tuple to the algorithm of Theorem~\ref{thm:DP}:
the output of both the algorithm of Theorem~\ref{thm:DP_intro} and Theorem~\ref{thm:DP}
is the same, we need only to verify the promise of Theorem~\ref{thm:DP}.

Let $F$ be an induced subgraph of $G$ of treewidth less than $k$.
By Proposition~\ref{prop:td2fillin}, there exists a minimal chordal completion $\mathcal{E}_F$
of $F$ such that every clique of $F+\mathcal{E}_F$ is of size at most $k$.
By Lemma~\ref{lem:FTV}, there exists a minimal chordal completion $\mathcal{E}_G$
of $G$ such that for every clique $\Omega$ of $G+\mathcal{E}_G$, the set 
$\Omega \cap V(F)$ is either empty or is a clique of $F + \mathcal{E}_F$.
In particular, if $(T,\beta)$ is the clique tree of $G+\mathcal{E}_G$ (from Proposition~\ref{prop:chordal_clique_tree}), then $|\beta(t) \cap V(F)| \leq k$ for every $t \in V(T)$, 
   so $(T,\beta_F)$ is a tree decomposition of $F$ of width less than $k$,
   where $\beta_F(t) = \beta(t) \cap V(F)$ for every $t \in V(T)$. 
Since $\beta(t)$ is a maximal clique of $G+\mathcal{E}_G$ for every $t \in V(T)$, 
by the assumptions of Theorem~\ref{thm:DP_intro}, $\coverfam$
contains an $F$-container for $\beta(t)$. 

This verifies the promise of Theorem~\ref{thm:DP} and thus completes the proof
of Theorem~\ref{thm:DP_intro}, assuming Theorem~\ref{thm:DP}.

\subsection{Proof of Theorem~\ref{thm:DP}}\label{ssec:DP_proof}
  
  Let us first give some intuition how the proof of Theorem~\ref{thm:DP} works and where it
  differs from the proofs of analogous statements proved by Fomin and Villanger~\cite{FominV10},
  and Fomin, Todinca, and Villanger~\cite{FominTV15}. 
  Fix a solution $F$ and a tree decomposition $(T,\beta)$ as in the theorem statement. 

  In~\cite{FominV10,FominTV15}, we are given a family $\mathcal{B}$ that contains all \emph{bags}
  of the tree decomposition $(T,\beta)$.
  The dynamic programming state consists of a set $B \in \mathcal{B}$,
  a set $Q \subseteq B$ of size at most $k$, and a component $D \in \cc(G-B)$.
  The dynamic programming algorithm computes a partial solution $\dpres(B,Q,D) \subseteq D$
  that is intended to fit to solutions $F'$ with $V(F') \cap B = Q$.
  That is, we aim at achieving $\dpres(B,Q,D) = D \cap V(F)$ whenever
  $B = \beta(t)$ for some $t \in V(T)$ and $Q = B \cap V(F)$. 
  In one step of the dynamic programming algorithm, given $(B,Q,D)$, the algorithm 
  tries all possibilities for $B' \in \mathcal{B}$ and $Q' \subseteq B'$ of size at most $k$.
  For fixed $B'$ and $Q'$, the algorithm assembles a candidate for $\dpres(B,Q,D)$
  from entries $\dpres(B',Q',D')$ for every $D' \in \cc(G-B')$ where $D \cap D' \neq \emptyset$.
  Whenever indeed $B = \beta(t)$ and $Q = B \cap V(F)$ for some $t \in V(T)$,
  we aim at obtaining the correct solution
  when using $B' = \beta(t')$ and $Q' = \beta(t') \cap V(F)$ for $t'$ being
  the neighbor of $t$ in the component of $T-\{t\}$ whose bags contain all vertices of $D$.
  
  In our algorithm, we given a list $\coverfam$ that contains only \emph{containers} for bags $\beta(t)$,
  not the bags exactly. The difficulty in the above approach 
  appears where the container $A$ of $\beta(t)$ is ``much larger'' than the container $A'$ of $\beta(t')$
  and for a number of components $D' \in \cc(G-A')$ we have both $D' \cap D \neq \emptyset$
  and $D' \cap A \neq \emptyset$ (which cannot happen in the setting of~\cite{FominV10,FominTV15}).
  Then, when the optimum solution is not unique, optimum partial solutions for
  states $(A',Q',D')$ may intersect $A$ outside $\beta(t)$, causing inconsistencies. 

  The main trick in our proof is to canonize the solution first to a lexicographically-minimum
  solution. This removes ambiguities in a way that can be decided on the level of partial
  solutions for a fixed state $(A,Q,D)$.

  \medskip
  Let us proceed to the main proof. 
  As promised, we start with some canonization definitions. 
	The \emph{lexicographic order} on subsets of $V(G)$ is defined as follows.
	We order the vertices of $V(G)$ arbitrarily as $\{v_1,v_2,\ldots,v_n\}$ where $n = |V(G)|$
	and with a set $B \subseteq V(G)$ we associate a $\{0,1\}$-vector $\iota_B$ of length $n$
	with $\iota_B[i] = 1$ if and only if $v_i \in B$, for $i \in [n]$. 
	For two subsets $B_1,B_2 \subseteq V(G)$, we have that $B_1$ is lexicographically earlier
	than $B_2$, $B_1 \lexlt B_2$ if $\iota_{B_1}$ is lexicographically earlier than $\iota_{B_2}$.
	Lexicographic order allows us to define an order $\indlt$ on induced subgraphs of $G$.
	If $F_1$ and $F_2$ are two induced subgraphs of $G$, then $F_1 \indlt F_2$
  if $\weight(V(F_1)) > \weight(V(F_2))$
	or $\weight(V(F_1)) = \weight(V(F_2))$ and $V(F_1) \lexlt V(F_2)$. 
	That is, the $\indlt$-minimum induced subgraph of treewidth less than $k$
  is the lexicographically first of all maximum-weight induced subgraphs of treewidth less than $k$.
  Our algorithm will in fact return such a set.
	
	We immediately have the following property.
	\begin{lemma}\label{lem:lex}
		If $B_1,B_2 \subseteq V(G)$ and $X \subseteq V(G)$ such that $B_1 \setminus X = B_2 \setminus X$,
		but $B_1 \cap X \lexlt B_2 \cap X$, then $B_1 \lexlt B_2$.
		Consequently, if $B_1,B_2 \subseteq V(G)$ are two vertex sets and $X \subseteq V(G)$ is such
		that $B_1 \setminus X = B_2 \setminus X$ but $B_1 \cap X \indlt B_2 \cap X$, then
		$B_1 \indlt B_2$.
	\end{lemma}
	

	
	We start by defining the set of states of our dynamic programming algorithm. A \emph{state}
	is a tuple $(A,Q,D)$ where $A \in \coverfam$, $Q \subseteq A$
	is of size at most $k$, and $D$ is a connected component of $G-A$. 
	Let $\States$ be the set of states.
  A set $P \subseteq D$ is a \emph{feasible solution} to the state $(A,Q,D)$
  if $G[P \cup Q]$ admits a tree decomposition of width less than $k$
  with $Q$ being contained in one of the bags.

  Observe that one can verify in time $n^{\Oh(k)}$ whether $P$ is a feasible solution
  to $(A,Q,D)$ by applying the algorithm of Arnborg, Corneil, and Proskurowski~\cite{tw-nk}
  (that verifies if a given $n$-vertex graph has treewidth less than $k$ in 
   time $\Oh(n^{k+1})$) to the graph $G[P \cup Q]$ with $Q$ turned into a clique.
  
	For every state $(A,Q,D) \in \States$ the algorithm will compute a set $\dpres(A,Q,D)$
	that is a feasible solution to $(A,Q,D)$. 
  The algorithm initializes $\dpres(A,Q,D) := \emptyset$ for every state $(A,Q,D)$;
  note that $\emptyset$ is a feasible solution to every state due to the assumption $|Q| \leq k$.

  We will need the following observation.
  \begin{lemma}\label{lem:DP-glue}
  Let $A \in \coverfam$, $Q \subseteq A$ be of size at most $k$,
  let $\mathcal{D} \subseteq \cc(G-A)$, and let $(J_D)_{D \in \mathcal{D}}$ be
 such that $J_D$ is a feasible solution to $(A,Q,D)$ 
  for every $D \in \mathcal{D}$.
  Define
  $$F(A,Q,\mathcal{D},(J_D)_{D \in \mathcal{D}}) := Q \cup \bigcup_{D \in \mathcal{D}} J_D.$$
  Then, $G[F(A,Q,\mathcal{D},(J_D)_{D \in \mathcal{D}})]$ admits a tree decomposition
  of width less than $k$ with $Q$ contained in one of the bags.
  \end{lemma}
  \begin{proof}
  Fix $D \in \mathcal{D}$. 
  Since $J_D$ is a feasible solution for $(A,Q,D)$, there exists
  a tree decomposition $(T_D,\beta_D)$ of $G[Q \cup J_D]$ of width less than $k$
  with a node $t_D \in V(T_D)$ such that $Q \subseteq \beta_D(t_D)$. 

  Construct a tree decomposition $(T,\beta)$ of $G[F(A,Q,\mathcal{D},(J_D)_{D \in \mathcal{D}})]$ as follows.
  First, let $T$ be obtained by taking a disjoint union of all trees $T_D$, for $D \in \mathcal{D}$,
  and adding a new node $t$, which is adjacent to $t_D$ for every $D \in \mathcal{D}$. 
  Second, define $\beta$ to be the union of all $\beta_D$ for $D \in \mathcal{D}$, and additionally
  $\beta(t) = Q$. 
  Then, $(T,\beta)$ is a tree decomposition of $G[F(A,Q,\mathcal{D},(J_D)_{D \in \mathcal{D}})]$
  of width less than $k$ with $\beta(t) = Q$, as desired.
  \cqed\end{proof}

	Let $F$ be the $\indlt$-minimum induced subgraph of $G$ of treewidth less than $k$.
	Let $(T,\beta)$ be the tree decomposition promised for $F$ in the theorem statement.
	By standard arguments, we can assume that $|E(T)| \leq |V(G)|$. Indeed, if there
	is an edge $t_1t_2 \in E(T)$ with $\beta(t_1) \subseteq \beta(t_2)$, we can contract the edge
	$t_1t_2$, keeping $\beta(t_2)$ as the bag associated to the resulting node. 
  It is straightforward to verify that such a contraction does not break the promised 
  properties of $(T,\beta)$. 
  If no such contraction
	is possible, root $T$ at an arbitrary node and observe that for every edge $t_1t_2$
	with $t_2$ being the parent and $t_1$ being the child,
	there is at least one vertex in $\beta(t_1) \setminus \beta(t_2)$ and every vertex
	of $V(G)$ can be an element of $\beta(t_1) \setminus \beta(t_2)$ for at most one
	pair $(t_1,t_2)$ where $t_1$ is a child of $t_2$. Thus, there are at most $|V(G)|$
	edges of $T$.
	
	For every $t \in V(T)$, let $A_t \in \coverfam$ be the container promised in the theorem statement,
	that is, $\beta(t) \subseteq A_t$ while $A_t \cap V(F) = \beta(t) \cap V(F)$.
  In particular, this implies that $|A_t \cap V(F)| \leq k$ for every $t \in V(T)$,
  so $(A_t,A_t \cap V(F),D) \in \States$ for every $t \in V(T)$ and $D \in \cc(G-A_t)$.
  
  We observe now the following straightforward corollary of the properties of a tree 
  decomposition.
  \begin{lemma}\label{lem:DP-tD}
  For every $t \in V(T)$ and $D \in \cc(G-A_t)$ there exists a unique neighbor $t_D$
  of $t$ in $T$ such that the vertices of $D$ appear only
  in bags in the component $T_D$ of $T-\{tt_D\}$ that contains $t_D$.
  \end{lemma}
  By the choice of $(T,\beta)$, the decomposition $(T,\beta_F)$ is a tree decomposition
  of $F$ of width less than $k$, where $\beta_F(t) = \beta(t) \cap V(F)$ for every $t \in V(T)$.
  
  Fix $t \in V(T)$ and $D \in \cc(G-A_t)$.
  Let $Q = V(F) \cap A_t$; since $A_t$ is an $F$-container for $\beta(t)$, we know that
  $Q = \beta_F(t)$. 
  Let $t_D$ and $T_D$ be as in Lemma~\ref{lem:DP-tD} for $t$ and $D$.
  Let $T_D'$ be obtained from the tree $T_D$ by adding the vertex $t$ and the edge $tt_D$.
  In other words, $T_D'$ is the subtree of $T$ induced by $V(T_D) \cup \{t\}$.
  Let $\beta_{F,t,D}$ be defined as $\beta_{F,t,D}(t') := \beta_F(t') \cap (D \cup A_t)$ 
  for all $t' \in V(T_D')$. 
  Then, $(T_D',\beta_{F,t,D})$ is a tree decomposition of $F[A_t \cup D]$ of width less than
  $k$, satisfying $Q = \beta_{F,t,D}(t)$. 
  Hence, $D \cap V(F)$ is a feasible solution to $(A_t,A_t \cap V(F),D)$. 

	Furthermore, Lemma~\ref{lem:lex} implies
	that the set $D \cap V(F)$ is $\indlt$-minimum feasible solution to $(A_t,A_t \cap V(F),D)$.
	Indeed, if there were a set $J \indlt \left( D \cap V(F) \right)$ that is also a feasible solution
  to $(A_t,A_t \cap V(F),D)$, then
  $F' := G[(V(F) \setminus D) \cup J]$ would also be of treewidth
 less than $k$ (thanks to Lemma~\ref{lem:DP-glue})
  and $V(F') \indlt V(F)$, contradicting the choice of $F$.
  
  We will prove that our algorithm actually finds $D \cap V(F)$
  as a feasible solution for every $t \in V(T)$ and $D \in \cc(G-A_t)$.
  That is, we will prove that in the end the algorithm attains the following property.
	\begin{equation}\label{eq:dp1}
	\dpres(A_t,A_t \cap V(F),D) = D \cap V(F)\qquad\qquad\mathrm{for\ every\ }t \in V(T)\mathrm{\ and\ }D \in \cc(G-A_t).
	\end{equation}
	Assume for the moment that the values $\dpres(\cdot)$ are computed such that~\eqref{eq:dp1}
	is satisfied. We show how to conclude. Iterate over all sets $A \in \coverfam$
	and sets $Q \subseteq A$ of size at most $k$. For every pair $(A,Q)$
	compute 
  $$F_{A,Q} := F(A,Q,\cc(G-A),(\dpres(A,Q,D))_{D \in \cc(G-A)}).$$

  Lemma~\ref{lem:DP-glue} asserts that $G[F_{A,Q}]$ is of treewidth less than $k$.
	Our algorithm returns the $\indlt$-minimum set among all considered sets $F_{A,Q}$.
	Clearly, given the values $\dpres(\cdot)$, choosing such $F_{A,Q}$ can be done
	in time $|\coverfam| \cdot |V(G)|^{\Oh(k)}$. Furthermore, for every $t \in V(T)$ there is an interation where
	the algorithm considers the pair $(A_t,A_t \cap V(F))$ and then~\eqref{eq:dp1} ensures that
	$F_{A_t,A_t \cap V(F)} = V(F)$. Thus, the algorithm returns $V(F)$.
	It remains to show how to compute the values $\dpres(\cdot)$ so that the property~\eqref{eq:dp1}
  is satisfied. 
	
	Recall that the algorithm initializes $\dpres(A,Q,D) := \emptyset$ for every $(A,Q,D) \in \States$.
	The algorithm performs $|V(G)|$ rounds. In each round, the algorithm inspects
	every state $(A,Q,D) \in \States$ and performs the following computation. 
	It iterates over every pair $(A',Q')$, where $A' \in \coverfam$ and $Q' \subseteq A'$
	is of size at most $k$, such that $Q \cap (A \cap A') = Q' \cap (A \cap A')$.
	For a fixed pair $(A',Q')$, let 
	$$\mathcal{D} := \{D' \in \cc(G-A')~|~D' \cap D \neq \emptyset\}.$$
	The algorithm inspects all values $\dpres(A',Q',D')$ for $D' \in \mathcal{D}$
  and computes 
  $$J := D \cap F(A',Q',\mathcal{D},(\dpres(A',Q',D'))_{D' \in \mathcal{D}}).$$
	If $J$ is a feasible solution to $(A,Q,D)$ and $J \indlt \dpres(A,Q,D)$,
  then the algorithm updates the value $\dpres(A,Q,D)$ by setting  $\dpres(A,Q,D) := J$. 
	We shall later refer to the above step as \emph{considering $J$ as a candidate for $\dpres(A,Q,D)$}.

	Clearly, the algorithm runs in time $|\coverfam|^2 |V(G)|^{\Oh(k)}$.
	It remains to show the property~\eqref{eq:dp1}.

  Fix $t \in V(T)$ and $D \in \cc(G-A_t)$.
	Since $D \cap V(F)$ is the $\indlt$-minimum feasible solution to $(A_t,A_t \cap V(F), D)$, 
  if at some moment the algorithm considers $D \cap V(F)$ as a candidate value
	for $\dpres(A_t,A_t \cap V(F),D)$, then it sets $\dpres(A_t,A_t \cap V(F),D) := D \cap V(F)$
	and never changes it later.
  Thus, it suffices to show that 
  the set $D \cap V(F)$ is at least once considered as a candidate for $\dpres(A_t,A_t \cap V(F),D)$.

	For a pair $(t_1,t_2)$ of adjacent nodes of $T$, the \emph{depth} of $(t_1,t_2)$ is the maximum
	number of edges on a simple path in $T$ that starts in $t_1$ and has $t_2$ as a second vertex.
  Let $t_D$ and $T_D$ be as in Lemma~\ref{lem:DP-tD} for $t$ and $D$.
	Let $d$ be the depth of $(t,t_D)$.
	We will show by induction on the depth of $(t,t_D)$ that $\dpres(A_t,A_t \cap V(F),D) = D \cap V(F)$ after
	$d$ rounds.
	
	To this end, we show that in $d$-th round we consider $J = D \cap V(F)$
	for the pair $(A',Q') = (A_{t_D}, A_{t_D} \cap V(F))$. 
  Clearly, $(A_t \cap V(F)) \cap (A_t \cap A_{t_D}) = (A_{t_D} \cap V(F)) \cap (A_t \cap A_{t_D})$,
  so the pair $(A',Q') = (A_{t_D}, A_{t_D} \cap V(F))$ is considered by the algorithm
  while iterating over pairs $(A',Q')$ for the state $(A_t,A_t \cap V(F), D)$. 
	Recall that
	$$\mathcal{D} = \{D' \in \cc(G-A_{t_D})~|~D' \cap D \neq \emptyset\}.$$
	
	From the properties of a tree decomposition we infer the following.
	\begin{lemma}\label{lem:progress}
		For every $D' \in \mathcal{D}$ there exists a neighbor $s_{D'}$ of $t_D$
		distinct from $t$ such that all vertices of $D'$
		lie only in bags of the component of $T-\{t_Ds_{D'}\}$ that contains $s_{D'}$.
	\end{lemma}
	\begin{proof}
		Since $\beta(t_D) \subseteq A_{t_D}$, for every $D' \in \cc(G-A_{t_D})$ 
		there exists a neighbor $s_{D'}$ of $t_D$ such that 
		all vertices of $D'$ lie only in bags of the component of $T-\{t_Ds_{D'}\}$ that contains $s_{D'}$.
		The crux is to show that if $D' \in \mathcal{D}$, then $s_{D'} \neq t$.
		
		Pick $v \in D' \cap D$. There exists a node $s \in V(T)$ with $v \in \beta(s)$. 
		By the choice of $t_D$, the node $s$ lies in the component of $T-\{tt_D\}$ that contains $t_D$.
		By the choice of $s_{D'}$, the node $s$ lies in the component of $T-\{t_Ds_{D'}\}$ that contains $s_{D'}$.
		Hence, $t = s_{D'}$ would give a contradiction. This completes the proof.
		\cqed\end{proof}
	Observe that for every neighbor $s$ of $t_D$ that is distinct from $t$,
	the depth of $(t_D,s)$ is strictly smaller than the depth of $(t,t_D)$.
	Consequently, by the inductive hypothesis, $\dpres(A_{t_D}, A_{t_D} \cap V(F), D') = D' \cap V(F)$
	for every $D' \in \mathcal{D}$. 
  Thus, the algorithm considers as a candidate for $\dpres(A_t,A_t \cap V(F), D)$ the value
	\begin{align*}
	J &= D \cap \left((A_{t_D} \cap V(F)) \cup \bigcup_{D' \in \mathcal{D}} \dpres(A_{t_D}, A_{t_D} \cap V(F),D')\right)\\
	&= D \cap \left((A_{t_D} \cap V(F)) \cup \bigcup \left\{D' \cap V(F)~|~D' \in \cc(G-A_{t_D}) \wedge D' \cap D \neq \emptyset\right\}\right)\\
	&= D \cap V(F).
	\end{align*}
	Hence, $\dpres(A_t,A_t \cap V(F),D) = D \cap V(F)$ after $d$ rounds of the algorithm.
	This completes the proof of property~\eqref{eq:dp1} and thus of Theorem~\ref{thm:DP}.

\section{Conclusion}\label{sec:conclusion}
In this paper, we modify the dynamic programming algorithm in the framework
of potential maximal cliques to take as input a set of containers of potential maximal cliques. 
We apply it to the class $\ourclass$ that contains both long-hole-free graphs and $P_5$-free graphs.
We hope that the method of containers will find applications in other scenarios as well.

We would like to discuss here three directions of generalizations of Theorem~\ref{thm:DP_intro}.
Recall the requirement of the theorem that for every induced subgraph $F$ of $G$ of treewidth
less than $k$ and every potential maximal clique $\pmc$ of $G$, the supplied family
$\coverfam$ contains an $F$-container for $\Omega$.

\paragraph{Allowing $\Oh(1)$ extra vertices of the solution in a container.}
In the first direction, let us focus on the requirement
$A \cap V(F) = \pmc \cap V(F)$ for the set $A$ to be an $F$-container for $\pmc$.
We observe that this requirement can be easily generalized to allow $A$ to contain a constant
number of vertices of $F$ that are not in $\pmc$.
More formally, for an integer $p$ and an induced subgraph $F$ of $G$, 
we say that $A \subseteq V(G)$ is an \emph{$(F,p)$-container} for $\pmc \subseteq V(G)$
if $\pmc \subseteq A$ and $|(A \setminus \pmc) \cap V(F)| \leq p$.
In particular, an $(F,0)$-container is an $F$-container.

Assume that we can enumerate a family $\coverfam$ with only the promise that
$\coverfam$ contains an $(F, p)$-container for $\pmc$ for every $F$ and $\pmc$
as in Theorem~\ref{thm:DP_intro}.
Then, the family
$$\coverfam' := \left\{A \setminus B~|~A \in \coverfam \wedge B \subseteq A \wedge |B| \leq p\right\}$$
is of size $\Oh(|\coverfam| n^p)$ and contains an $F$-container
for every $F$ and $\Omega$.

\paragraph{Enumerating containers for only selected PMCs.}
In the second direction, let us focus on the necessity to enumerate in $\coverfam$ a container for \emph{every} PMC. 
The main insight of the work of Lokshtanov, Vatshelle, and Villanger~\cite{LokshtanovVV14} is
to enumerate only some PMCs, guaranteeing that for the sought solution $I$ there exists
a minimal chordal completion that does not add any edge incident with $I$ and all maximal
cliques of that completion are enumerated. 
An astute reader can notice that the statement of Theorem~\ref{thm:DP}, a technical statement behind Theorem~\ref{thm:DP_intro},
requires only to list containers for bags of the promised tree decomposition $(T,\beta)$ of $G$ for any feasible solution $F$. 
Furthermore, in the proof of Theorem~\ref{thm:DP_intro}, we use only containers for bags of the decomposition $(T,\beta)$ for the
$\indlt$-minimum solution $F$ (i.e., lexicographically-minimum solution of maximum weight).
Hence, we can state the following generalization of Theorem~\ref{thm:DP_intro}.
\begin{theorem}\label{thm:DP_intro_gen}
Assume we are given a graph $G$ with weight function $\weight : V(G) \to \N$, 
a family $\coverfam$ of subsets of $V(G)$, and a positive integer $k$ with the following promise:
\begin{displayquote}
For every induced subgraph $F$ of $G$ of treewidth less than $k$
there exists a minimal chordal completion $\mathcal{E}$ of $G$
such that
\begin{itemize}
\item every clique of $(G+\mathcal{E})[V(F)]$ is of size at most $k$, and
\item for every maximal clique $\pmc$ of $G+\mathcal{E}$,
  $\coverfam$ contains an $F$-container for $\pmc$.
\end{itemize}
\end{displayquote}
\noindent Then, one can in time $|\coverfam|^2 |V(G)|^{\Oh(k)}$
find a maximum-weight induced subgraph of $(G,\weight)$ of treewidth less than $k$.
\end{theorem}

Lemma~\ref{lem:FTV}, originating in~\cite{FominV10,FominTV15}, is the crucial observation
allowing us to go from the world of tree decompositions in Theorem~\ref{thm:DP}
to the world of minimal chordal completions in Theorem~\ref{thm:DP_intro}.
For the special case of \textsc{MWIS} (i.e., $k=1$ in Theorem~\ref{thm:DP_intro}),
Lemma~\ref{lem:FTV} boils down exactly to an existence of a minimal chordal completion
of $G$ that does not add any edge incident to the $\indlt$-minimum solution $F$. 
Taking into account also the discussion in the previous paragraphs,
we can state the following variant of Theorem~\ref{thm:DP_intro_gen}, tailored
for \textsc{MWIS}.
\begin{theorem}
\label{thm:DP_MWIS}
Assume we are given a graph $G$ with weight function $\weight : V(G) \to \N$,
a family $\coverfam$ of subsets of $V(G)$, and an integer $p$ with the following promise:
\begin{displayquote}
For every maximal independent set $I$ of $G$
there exists a minimal chordal completion $\mathcal{E}$ of $G$ such that
\begin{itemize}
\item $\mathcal{E}$ does not contain any edge incident with $I$, and 
\item for every maximal clique $\Omega$ of $G+\mathcal{E}$,
$\coverfam$ contains an $(I,p)$-container for $\Omega$.
\end{itemize}
\end{displayquote}
\noindent Then, one can in time $|\coverfam|^2 |V(G)|^{\Oh(p)}$
find a maximum-weight independent set in $(G,\weight)$.
\end{theorem}
That is, Theorem~\ref{thm:DP_MWIS}, being in fact a special case of Theorem~\ref{thm:DP} for $k=1$,
generalizes Theorem~\ref{thm:PMC_nonexhaustive_list} to containers.

\paragraph{Counting Monadic Second Order logic.}
In the third direction, we focus the use of Counting Monadic Second Order logic (CMSO), as in the work of Fomin, Todinca, and Villanger~\cite{FominTV15}.
The syntax of CMSO consists of basic boolean operations, vertex, edge, vertex set, and edge sets variables, and equality, containment, and incidence relations. 
Fomin, Todinca, and Villanger~\cite{FominTV15} considered the following problem for fixed CMSO formula $\phi$ with one free vertex set variable and an integer $k$: given a graph $G$, find
a pair $(F,X)$ maximizing $|X|$ such that $F$ is an induced subgraph of $G$ of treewidth less than $k$, $X \subseteq V(F)$, and $(F,X)$ satisfy $\phi$. 
They show that the problem can be solved in time polynomial in the size of $G$ and the number of PMCs in $G$, even if the input is equipped with vertex weights and we aim at maximizing the weight
of $X$. Note that this (weighted) problem generalizes the problem considered in Theorem~\ref{thm:DP_intro} by taking $\phi$ that requires $X = V(F)$. 

We observe that the same use of CMSO can smoothly and effortlessly be embedded into Theorems~\ref{thm:DP_intro} and~\ref{thm:DP}. 
That is, instead of asking for induced subgraph $F$ of treewidth less than $k$ maximizing the weight of $V(F)$, we can fix a CMSO formula $\phi$ as above and ask for a pair $(F,X)$
maximizing the weight of $X$ such that $F$ is an induced subgraph of $G$ of treewidth less than $k$, $X \subseteq V(F)$, and $(F,X)$ satisfy $\phi$. 
Then, the running time bound would be multiplied by a term depending only on $\phi$ and $k$:
\begin{theorem}\label{thm:DP_CMSO}
Assume we are given a graph $G$ with weight function $\weight : V(G) \to \N$, 
a family $\coverfam$ of subsets of $V(G)$, a positive integer $k$,
and a CMSO formula $\phi$ with one free vertex set variable,
 with the following promise:
\begin{displayquote}
For every induced subgraph $F$ of $G$ of treewidth less than $k$
and every potential maximal clique $\Omega$ of $G$,
if $|V(F) \cap \Omega| \leq k$, then $\coverfam$ contains an
$F$-container for $\Omega$.
\end{displayquote}
\noindent Then, one can in time $C(\phi, k) \cdot |\coverfam|^2 |V(G)|^{\Oh(k)}$
find a pair $(F,X)$ maximizing the weight of $X$ subject to the following constraints:
$F$ is an induced subgraph of $G$ of treewidth less than $k$,
$X \subseteq V(F)$, and $\phi$ is satisfied on $(F,X)$.
Here, $C(\phi,k)$ is a constant depending only on $\phi$ and $k$.
\end{theorem}

We refer to~\cite{FominTV15} for examples of problems expressible by this formalism. 

The work of~\cite{FominTV15} relies on previous framework by Borie, Parker, and Tovey~\cite{BoriePT92} to handle the CMSO property $\phi$. 
The key property of CMSO formulae is that they define \emph{regular} properties:
in our setting, given a pair $(F,X)$ with $X \subseteq V(F)$, a vertex separator $Q$ of $F$ of size at most $k$,
and a component $P$ of $G-Q$, there is only a bounded in $k$ and the size of $\phi$ number of potential ``types of partial behavior'' of $\phi$ on the tuple $(F[Q \cup P], Q, X \cap (Q \cup P))$.
We refer to~\cite{FominTV15} for a gentle introduction and precise definitions. 

In the dynamic programming algorithm inside the proof of Theorem~\ref{thm:DP}, handling a CMSO requirement $\phi$ can be done exactly in the same way as it is done
in the analogous dynamic programming algorithm in~\cite{FominTV15}. Recall that the state of the algorithm consists of a container $A \in \coverfam$, a set $Q \subseteq A$ of size at most $k$
(intended intersection of the solution with $A$) and a component $D \in \cc(G-A)$. The state seeks to extend the solution into $D$: a feasible solution to $(A,Q,D)$ is a set $P \subseteq D$
such that $G[Q \cup P]$ admits a tree decomposition of width less than $k$ with $Q$ contained
in one bag. 
With the CMSO requirement $\phi$, we need to extend the dynamic programming state
to a tuple $(A,Q,Q_X,c,D)$, where $Q_X \subseteq Q$ is the intended intersection of the set $X$
with $Q$ and $c$ is the $\phi$-type of a sought feasible solution inside $D$.
That is, now a partial solution is a pair $(P,Y)$ with $Y \subseteq P \subseteq D$
such that $G[Q \cup P]$ admits a tree decomposition of width less than $k$ with $Q$ contained
in one bag and the tuple $(G[Q \cup P], Q, Q_X \cup Y)$ has $\phi$-type $c$;
partial solutions are compared by the weight of $Y$.

We decided to omit the above generalization in the proof of Theorem~\ref{thm:DP} for the sake
of clarity of the arguments. The above generalization is a straightforward application 
of the techniques of~\cite{FominTV15} that would bring here a large definitional overhead
without bringing any new insight. 


\bibliographystyle{abbrv}

\bibliography{../references}

\end{document}